\documentclass[conference,11pt,letterpaper,onecolumn]{IEEEtran}

\usepackage{cite}

\usepackage{amsmath}
\usepackage{array}

\usepackage[pdftex]{graphicx}

\usepackage{amsfonts,amssymb,amsthm,bbm,verbatim}
\usepackage[usenames,dvipsnames]{color}
\usepackage[margin=1in]{geometry}
\usepackage[]{inputenc}
\usepackage{hyperref}
\hypersetup{colorlinks = true, linkcolor = blue, citecolor = blue, urlcolor  = blue}
\usepackage[all]{hypcap} % for going to the top of an image when a figure reference is clicked

\makeatletter % for matrices with labels
\def\bbordermatrix#1{\begingroup \m@th
  \@tempdima 4.75\p@
  \setbox\z@\vbox{%
    \def\cr{\crcr\noalign{\kern2\p@\global\let\cr\endline}}%
    \ialign{$##$\hfil\kern2\p@\kern\@tempdima&\thinspace\hfil$##$\hfil
      &&\quad\hfil$##$\hfil\crcr
      \omit\strut\hfil\crcr\noalign{\kern-\baselineskip}%
      #1\crcr\omit\strut\cr}}%
  \setbox\tw@\vbox{\unvcopy\z@\global\setbox\@ne\lastbox}%
  \setbox\tw@\hbox{\unhbox\@ne\unskip\global\setbox\@ne\lastbox}%
  \setbox\tw@\hbox{$\kern\wd\@ne\kern-\@tempdima\left[\kern-\wd\@ne
    \global\setbox\@ne\vbox{\box\@ne\kern2\p@}%
    \vcenter{\kern-\ht\@ne\unvbox\z@\kern-\baselineskip}\,\right]$}%
  \null\;\vbox{\kern\ht\@ne\box\tw@}\endgroup}
\makeatother

\newcommand{\Mod}[1]{\ (\mathrm{mod}\ #1)} % for modulo

\def\C{{\mathcal{C}}}
\def\X{{\mathcal{X}}}
\def\Y{{\mathcal{Y}}}

\def\F{{\mathcal{F}}}
\def\P{{\mathbb{P}}}

\def\E{{\mathbb{E}}}
\def\VAR{{\mathbb{V}\mathbb{A}\mathbb{R}}}

\def\R{{\mathbb{R}}}
\def\N{{\mathbb{N}}}
\def\1{{\textbf{1}}}
\def\I{{\mathbbm{1}}}

\DeclareMathOperator*{\br}{br}

\def\maj{{\textsf{\small majority}}}
\def\Ber{{\textsf{\small Bernoulli}}}

\newtheorem{theorem}{Theorem}
\newtheorem{lemma}{Lemma}
\newtheorem{proposition}{Proposition}
\newtheorem{corollary}{Corollary}

\renewcommand{\proofname}{\bfseries \emph{Proof}}

\IEEEoverridecommandlockouts

\begin{document}

\title{Broadcasting on Bounded Degree DAGs}

\begin{comment}
\author{\IEEEauthorblockN{Anuran Makur}
\IEEEauthorblockA{MIT, Department of EECS\\
Email: \href{mailto:a_makur@mit.edu}{a\_makur@mit.edu}}
\and
\IEEEauthorblockN{Elchanan Mossel}
\IEEEauthorblockA{MIT, Department of Mathematics\\
Email: \href{mailto:elmos@mit.edu}{elmos@mit.edu}}
\and
\IEEEauthorblockN{Yury Polyanskiy}
\IEEEauthorblockA{MIT, Department of EECS\\
Email: \href{mailto:yp@mit.edu}{yp@mit.edu}}}
\end{comment}

\author{Anuran Makur$^{\ast}$, Elchanan Mossel$^{\dagger}$, and Yury Polyanskiy$^{\ast}$
\thanks{$^{\ast}$A. Makur and Y. Polyanskiy are with the Department of Electrical Engineering and Computer Science, Massachusetts Institute of Technology, Cambridge, MA 02139, USA (e-mail: \href{mailto:a_makur@mit.edu}{a\_makur@mit.edu}; \href{mailto:yp@mit.edu}{yp@mit.edu}).

The research was supported in part by the Center for Science of Information (CSoI), an NSF Science and Technology Center, under grant agreement CCF-09-39370, by the NSF CAREER award CCF-12-53205, and by Schneider Electric, Lenovo Group (China) Limited and the Hong Kong Innovation and Technology Fund (ITS/066/17FP) under the HKUST-MIT Research Alliance Consortium.

$^{\dagger}$E. Mossel is with the Department of Mathematics, Massachusetts Institute of Technology, Cambridge, MA 02139, USA (e-mail: \href{mailto:elmos@mit.edu}{elmos@mit.edu}).

The research was supported in part by the NSF grants CCF-1665252 and DMS-1737944, and the DOD ONR grant N00014-17-1-2598.}}%

\maketitle

\thispagestyle{plain}
\pagestyle{plain}

\begin{abstract}
We study the following generalization of the well-known model of broadcasting on trees. Consider an infinite directed
acyclic graph (DAG) with a unique source node $X$. Let the collection of nodes at distance $k$ from $X$ be called the $k$th layer. At time zero, the source node is given a bit. At time $k\geq 1$, each node in the $(k-1)$th layer inspects its inputs and sends a bit to its descendants in the $k$th layer. Each bit is flipped with a probability of error $\delta \in \left(0,\frac{1}{2}\right)$ in the process of transmission. The goal is to be able to recover the original bit with probability of error better than $\frac{1}{2}$ from the values of all nodes at an arbitrarily deep layer $k$. 

Besides its natural broadcast interpretation, the DAG broadcast is a natural model of noisy computation. Some special cases of the model represent information flow in biological networks, and other cases represent noisy finite automata models. 

We show that there exist DAGs with bounded degree and layers of size $\omega(\log(k))$ that permit recovery provided $\delta$ is sufficiently small and find the critical $\delta$ for the DAGs constructed. Our result demonstrates a doubly-exponential advantage for storing a bit in bounded degree DAGs compared to trees. On the negative side, we show that if the DAG is a two-dimensional regular grid, then recovery is impossible for any $\delta \in \left(0,\frac{1}{2}\right)$ provided all nodes use either AND or XOR for their processing functions.
\end{abstract}

\tableofcontents
\hypersetup{linkcolor = red}

\section{Introduction and Main Results}

In this paper, we study a generalization of the well-known problem of broadcasting on trees \cite{BroadcastingonTrees}. In the broadcasting on trees problem, we are given a noisy tree $T$ whose nodes are Bernoulli random variables and edges are independent binary symmetric channels (BSCs) with common crossover probability $\delta \in \left(0,\frac{1}{2}\right)$. Given that the root is an unbiased random bit, the objective is to decode the bit at the root from the bits at the $k$th layer of the tree. The authors of \cite{BroadcastingonTrees} characterize the sharp threshold for when such reconstruction is possible: 
\begin{itemize}
\item If $(1-2\delta)^2 \br(T) < 1$, then the minimum probability of error in decoding tends to $\frac{1}{2}$ as $k \rightarrow \infty$, 
\item If $(1-2\delta)^2 \br(T) > 1$, then the minimum probability of error in decoding is bounded away from $\frac{1}{2}$ for all $k$, 
\end{itemize}
where $\br(T)$ denotes the branching number of the tree. A consequence of this result is that reconstruction is impossible for trees with sub-exponentially many vertices at each layer. Indeed, if $L_k$ denotes the number of vertices at layer $k$ and $\lim_{k \rightarrow \infty}{\log(L_k)/k} = 0$, then it is straightforward to show that $\br(T) \leq 1$, which in turn implies that $(1-2\delta)^2 \br(T) < 1$.

This result on reconstruction on trees generalizes results from statistical physics that hold for regular trees~\cite{BlRuZa:95}, and have had numerous extensions and further generalizations including~\cite{Ioffe:96b,Ioffe:96a,Mossel:98,Mossel:01,PemantlePeres:10,Sly:09a,Sly:09b,JansonMossel:04,BVVW:11}.
Reconstruction on trees plays a crucial role in understanding phylogenetic reconstruction, see e.g.~\cite{Mossel:03,Mossel:04a,DaMoRo:06,Roch:10}. 
It also plays a crucial role in understanding phase transitions for random constraint satisfaction problems, see e.g.~\cite{MezardMontanari:06,KMRSZ:07,GerschenfeldMontanari:07,MoReTe:11} and follow-up work. 

Instead of analyzing trees, we consider the problem of broadcasting on \textit{bounded degree directed acyclic graphs}
(DAGs). As in the setting of trees, all nodes in our graphs are Bernoulli random variables and all edges are independent BSCs. Furthermore, variables located at nodes with indegree $2$ or more are the values of a function on their noisy inputs. 

Notice that compared to the setting of trees, broadcasting on DAGs has two principal differences: (a) in trees, layer sizes scale exponentially in depth, while in DAGs they are polynomial; (b) in trees, the indegree of each node is $1$, while in DAGs each node has several incoming signals. The latter enables the possibility of information fusion at the nodes and our main goal is to understand whether the benefits of (b) overpower the harm of (a).

This paper contains two results. First, by a probabilistic argument, we demonstrate the existence of bounded degree DAGs with $L_k = \omega(\log(k))$ which permit recovery of the root bit for
sufficiently low $\delta$'s. This implies that in terms of economy of storing information, \textit{DAGs are
doubly-exponentially more efficient than trees}. Second, we show that no such recovery is possible on a two-dimensional (2D) grid if all intermediate nodes with indegree $2$ use logical AND as the processing function, or all use XOR as the processing function. (This leaves only NAND as the remaining symmetric processing function.)

\subsection{Motivation}
\label{Motivation}

The problem of broadcasting on trees is closely related to the problem of noisy computation~\cite{vonNeumann:56,EvansSchulman99}. Indeed it can be thought of in the following way: suppose we want to remember a bit in a noisy circuit of depth $k$. How big should the circuit be? Von Neumann~\cite{vonNeumann:56} asked this question assuming we take multiple clones of the original bit and recursively apply gates in order to reduce the noise. The broadcasting on trees perspective is to start from a single bit and repeatedly clone it so that one can recover it well from the nodes at depth $k$. The model we consider here again starts from a single bit but we are allowed to use bounded degree gates to reduce noise as well as to duplicate. This leads to much smaller circuits than the tree circuits. 

As mentioned earlier, the broadcasting process on trees plays a fundamental role in phylogenetic reconstruction. The positive results obtained here suggest it might be possible to reconstruct other biological networks, such as phylogenetic networks (see e.g.~\cite{HuRuSo:10}) or pedigrees (see e.g.~\cite{Thompson:86,SteelHein:06}), even if the growth of the network is very mild. It is interesting to explore if there are also connections between 
broadcasting on DAGs and random constraint satisfaction problems. Currently, we are not aware that such connections have been established. 

%Another motivation for this problem is to understand whether it is possible to propagate information in regular grids starting from the root--see Figure \ref{Figure: Grid} for a 2D example. Our \textit{conjecture} is that such propagation is possible for sufficiently low noise $\delta$ in $3$ and more dimensions, and impossible for a 2D grid regardless of the noise level and of the choice of processing functions. The conjecture resembles and is motivated by the positive rate conjecture on cellular automata \cite{Gray01}.

Another motivation for this problem is to understand whether it is possible to propagate information in regular grids starting from the root--see Figure \ref{Figure: Grid} for a 2D example. Our \textit{conjecture} is that such propagation is possible for sufficiently low noise $\delta$ in $3$ and more dimensions, and impossible for a 2D grid regardless of the noise level and of the choice of  processing function (which is the same for every node). The conjecture is inspired by the work on 1D cellular automata~\cite{Gray01}. Indeed, the existence of a 2D grid (with a choice of processing function) which remembers its initial state (bit) for infinite time would suggest the existence of non-ergodic infinite 1D cellular automata consisting of 2-input binary-state cells. Known constructions, however, require a lot more states~\cite{gacs2001reliable}, or are non-uniform in time and space~\cite{cirel1978reliable}.

In this paper, we take some first steps towards establishing this conjecture. The next few subsections formally define the random DAG and deterministic 2D grid models, and present our main results. After stating each result, we also provide a brief outline of the main technique or intuition used in the proof. The subsequent sections contain the proofs and auxiliary results.

\subsection{Random DAG Model}
\label{Random Grid Model}

A \textit{random DAG model} consists of an infinite DAG with fixed vertices that are Bernoulli ($\{0,1\}$-valued) random variables and randomly generated edges which are independent BSCs. We first define the vertex structure of this model, where each vertex is identified with the corresponding random variable. Let the root random variable be $X_{0,0} \sim \Ber\!\left(\frac{1}{2}\right)$. Furthermore, we define $X_k = (X_{k,0},\dots,X_{k,L_k-1})$ as the vector of node random variables at distance (i.e. length of shortest path) $k \in \N \triangleq \{0,1,2,\dots\}$ from the root, where $L_k \in \N$ denotes the number of nodes at distance $k$. In particular, we have $X_0 = (X_{0,0})$ and $L_0 = 1$. 

We next define the edge structure of the random DAG model. For any $k \in \N\backslash\!\{0\}$ and any $j \in [L_{k}] \triangleq \{0,\dots,L_{k} - 1\}$, we independently and uniformly select $d \in \N \backslash \! \{0\}$ vertices $X_{k-1,i_1},\dots,X_{k-1,i_d}$ from $X_{k-1}$ (i.e. $i_1,\dots,i_d$ are i.i.d. uniform on $[L_{k-1}]$), and then construct $d$ directed edges: $(X_{k-1,i_1},X_{k,j}),\dots,$ $(X_{k-1,i_d},X_{k,j})$. (Here, $i_1,\dots,i_d$ are independently chosen for each $X_{k,j}$.) This random process generates the underlying DAG structure. In the sequel, we will let $G$ be a random variable representing this underlying (infinite) random DAG, i.e. $G$ encodes the random configuration of the edges between the vertices.

To define a \textit{Bayesian network} (or  directed graphical model) on this random DAG, we fix some sequence of Boolean functions $f_{k}:\{0,1\}^d \rightarrow \{0,1\}$ for $k \in \N\backslash\!\{0\}$ (that depend on the level index $k$, but not on the realization of $G$), and some crossover probability $\delta \in \left(0,\frac{1}{2}\right)$ (since this is the interesting regime of $\delta$). Then, for any $k \in \N\backslash\!\{0\}$ and $j \in [L_{k}]$, given $i_1,\dots,i_d$ and $X_{k-1,i_1},\dots,X_{k-1,i_d}$, we define:\footnote{In this model, the Boolean processing function used at a node $X_{k,j}$ depends only on the level index $k$. A more general model can be defined where each node $X_{k,j}$ has its own Boolean processing function $f_{k,j}:\{0,1\}^d \rightarrow \{0,1\}$ for $k \in \N\backslash\!\{0\}$ and $j \in [L_{k}]$, but we will only analyze instances of the simpler model in this paper.}
\begin{equation}
\label{eq:propagation}
X_{k,j} = f_{k}(X_{k-1,i_1} \oplus Z_{k,j,1},\dots,X_{k-1,i_d}\oplus Z_{k,j,d})
\end{equation}
where $\oplus$ denotes addition modulo $2$, and $\{Z_{k,j,i} : k \in \N\backslash\!\{0\}, j \in [L_{k}], i \in \{1,\dots,d\}\}$ are i.i.d $\Ber(\delta)$ random variables that are independent of everything else. This means that each edge is a BSC with parameter $\delta$ (denoted $\textsf{\small BSC}(\delta)$). Moreover, \eqref{eq:propagation} characterizes the conditional distribution of $X_{k,j}$ given its parents.

Note that although we will analyze this model for convenience, as stated, our underlying graph is really a directed multigraph rather than a DAG, because we select the parents of a vertex with replacement. It is straightforward to construct an equivalent model where the underlying graph is truly a DAG. For each vertex $X_{k,j}$ with $k \in \N\backslash\!\{0\}$ and $j \in [L_{k}]$, we first construct $d$ intermediate parent vertices $\{X_{k,j}^i : i = 1,\dots,d\}$ that live between layers $k$ and $k-1$, where each $X_{k,j}^i$ has a single edge pointing to $X_{k,j}$. Then, for each $X_{k,j}^i$, we independently and uniformly select a vertex from layer $k-1$, and construct a directed edge from that vertex to $X_{k,j}^i$. This defines a valid (random) DAG. As a result, every realization of $G$ can be perceived as either a directed multigraph or its equivalent DAG. Furthermore, the Bayesian network on this true DAG is defined as follows: each $X_{k,j}$ is the output of a Boolean processing function $f_{k}$ with inputs $\{X_{k,j}^i : i = 1,\dots,d\}$, and each $X_{k,j}^i$ is the output of a BSC whose input is the unique parent of $X_{k,j}^i$ in layer $k-1$. 

Finally, we define the ``empirical probability of unity'' at level $k \in \N$ as:
\begin{equation}
\sigma_k \triangleq \frac{1}{L_k} \sum_{m = 0}^{L_k - 1}{X_{k,m}}
\end{equation}
where $\sigma_0 = X_{0,0}$ is just the root node. Observe that given $\sigma_{k-1} = \sigma$, $X_{k-1,i_1},\dots,X_{k-1,i_d}$ are i.i.d. $\Ber(\sigma)$, and as a result, $X_{k-1,i_1} \oplus Z_1,\dots,X_{k-1,i_d} \oplus Z_d$ are i.i.d. $\Ber(\sigma \star \delta)$, where $\sigma \star \delta \triangleq \sigma(1-\delta) + \delta(1-\sigma)$ is the convolution of $\sigma$ and $\delta$. Therefore, $X_{k,j}$ is the output of $f_{k}$ upon inputting a $d$-length i.i.d. $\Ber(\sigma \star \delta)$ string.

Under this setup, our objective is to determine whether or not the value at the root $\sigma_0 = X_{0,0}$ can be decoded from the observations $X_k$ as $k \rightarrow \infty$. Since $X_k$ is an exchangeable sequence of random variables given $\sigma_0$, for any $x_{0,0},x_{k,0},\dots,x_{k,L_k-1} \in \{0,1\}$ and any permutation $\pi$ of $[L_k]$, we have $P_{X_k|\sigma_0}(x_{k,0},\dots,x_{k,L_k-1}|x_{0,0}) = P_{X_k|\sigma_0}(x_{k,\pi(0)},\dots,x_{k,\pi(L_k-1)}|x_{0,0})$. Letting $\sigma = \frac{1}{L_k}\sum_{j = 0}^{L_k - 1}{x_{k,j}}$, we can factorize $P_{X_k|\sigma_0}$ as:
\begin{equation}
P_{X_k|\sigma_0}(x_{k,0},\dots,x_{k,L_k-1}|x_{0,0}) = {\binom{L_k}{L_k \sigma}}^{\! -1} P_{\sigma_k|\sigma_0}(\sigma|x_{0,0}) \, . 
\end{equation}
Using the Fisher-Neyman factorization theorem \cite[Theorem 3.6]{Statistics}, this implies that $\sigma_k$ is a \textit{sufficient statistic} of $X_k$ for performing inference about $\sigma_0$. Therefore, we restrict our attention to the Markov chain $\{\sigma_k : k \in \N\}$. Given $\sigma_k$, inferring the value of $\sigma_0$ is a binary hypothesis testing problem with minimum achievable probability of error:
\begin{equation}
\P\!\left(f_{\sf{ML}}^k(\sigma_k) \neq \sigma_0\right) = \frac{1}{2}\left(1-\left\|P_{\sigma_k}^+ - P_{\sigma_k}^-\right\|_{\sf{TV}}\right)
\end{equation}
where $f_{\sf{ML}}^k:\{m/L_k : m = 0,\dots,L_k\} \rightarrow \{0,1\}$ is the maximum likelihood (ML) decision rule at level $k$ in the absence of knowledge of the random DAG realization $G$, $P_{\sigma_k}^+$ and $P_{\sigma_k}^-$ are the conditional distributions of $\sigma_k$ given $\sigma_0 = 1$ and $\sigma_0 = 0$ respectively, and for any two probability measures $P$ and $Q$ on the same measurable space $(\Omega,\F)$, their \textit{total variation (TV) distance} is defined as:
\begin{equation}
\left\|P - Q\right\|_{\sf{TV}} \triangleq \sup_{A \in \F}{\left|P(A) - Q(A)\right|} \, .
\end{equation}
%(We refer readers to \cite[Theorem 6.3]{InfoTheoryNotes} and \cite[Chapter 4]{MarkovMixing} for various equivalent characterizations of TV distance.) 
We say that reconstruction of the root bit $\sigma_0$ is possible when:
\begin{equation}
\label{Eq:TV Reconstruction Possible}
\limsup_{k \rightarrow \infty}{\P\!\left(f_{\sf{ML}}^k(\sigma_k) \neq \sigma_0\right)} < \frac{1}{2} \quad \Leftrightarrow \quad \liminf_{k \rightarrow \infty}{\left\|P_{\sigma_k}^+ - P_{\sigma_k}^-\right\|_{\sf{TV}}} > 0 \, , 
\end{equation}
and is impossible when:
\begin{equation}
\label{Eq:TV Reconstruction Impossible}
\lim_{k \rightarrow \infty}{\P\!\left(f_{\sf{ML}}^k(\sigma_k) \neq \sigma_0\right)} = \frac{1}{2} \quad \Leftrightarrow \quad \lim_{k \rightarrow \infty}{\left\|P_{\sigma_k}^+ - P_{\sigma_k}^-\right\|_{\sf{TV}}} = 0 \, .
\end{equation}
In the sequel, to simplify our analysis when proving that reconstruction is possible, we will often use other (sub-optimal) decision rules rather than the ML decision rule. On the other hand, when proving that reconstruction is impossible, we will prove the stronger impossibility result:
\begin{equation}
\label{Eq:Strong TV Reconstruction Impossible}
\lim_{k \rightarrow \infty}{\E\!\left[\P\!\left(f_{\sf{ML}}^k(\sigma_k,G) \neq \sigma_0 \middle| G\right)\right]} = \frac{1}{2} \quad \Leftrightarrow \quad \lim_{k \rightarrow \infty}{\E\!\left[\left\|P_{\sigma_k|G}^+ - P_{\sigma_k|G}^-\right\|_{\sf{TV}}\right]} = 0 
\end{equation}
where $f_{\sf{ML}}^k(\cdot,G):\{m/L_k : m = 0,\dots,L_k\} \rightarrow \{0,1\}$ is the ML decision rule at level $k$ given knowledge of the random DAG realization $G$ (based on $\sigma_k$, not the full $k$-layer state $X_k$), and $P_{\sigma_k|G}^+$ and $P_{\sigma_k|G}^-$ denote the conditional distributions of $\sigma_k$ given $\{\sigma_0 = 1,G\}$ and $\{\sigma_0 = 0,G\}$, respectively. Note that applying Jensen's inequality to the TV distance condition in \eqref{Eq:Strong TV Reconstruction Impossible} establishes the weaker impossibility result in \eqref{Eq:TV Reconstruction Impossible}.

\subsection{Results on Random DAG Models}
\label{Results on Random DAG Models}

We prove two main results on the random DAG model. The first considers the setting where the indegree of each node (except the root) is $d = 3$. In this scenario, taking a majority vote of the inputs at each node intuitively appears to have good ``local error correction'' properties. So, we fix all Boolean functions in the random DAG model to be the \textit{majority} rule, and prove that this model exhibits a phase transition phenomenon around a critical threshold $\delta_{\sf{maj}} \triangleq \frac{1}{6}$. Indeed, the theorem below illustrates that for $\delta < \delta_{\sf{maj}}$, the majority decision rule $\hat{S}_k \triangleq \I\!\left\{\sigma_k \geq \frac{1}{2}\right\}$ can asymptotically decode $\sigma_0$, but for $\delta > \delta_{\sf{maj}}$, the ML decision rule cannot asymptotically decode $\sigma_0$. 

\begin{theorem}[Phase Transition in Random DAG Model with Majority Rule Processing]
\label{Thm:Phase Transition in Random Grid with Majority Rule Processing}
For a random DAG model with $d = 3$ and majority processing functions, the following phase transition phenomenon occurs around $\delta_{\sf{maj}}$:
\begin{enumerate}
\item If $\delta \in \left(0,\delta_{\sf{maj}}\right)$, and the number of vertices per level satisfies $L_k = \omega(\log(k))$, then reconstruction is possible in the sense that:
%\footnote{Throughout this paper, we will use Bachmann-Landau little-$o$, little-$\omega$, big-$O$, and big-$\Omega$ notation. In particular, given positive functions $f:\N \rightarrow \R$ and $g:\N \rightarrow \R$, we write $f(k) = o(g(k))$ if and only if $\lim_{k \rightarrow\infty}{f(k)/g(k)} = 0$, we write $f(k) = \omega(g(k))$ if and only if $\lim_{k \rightarrow\infty}{g(k)/f(k)} = 0$, we write $f(k) = O(g(k))$ if and only if $\limsup_{k \rightarrow\infty}{f(k)/g(k)} < +\infty$, and we write $f(k) = \Omega(g(k))$ if and only if $\liminf_{k \rightarrow\infty}{f(k)/g(k)} > 0$. Moreover, any other parameters $f$ and $g$ depend on are held constant during these limiting processes.} 
$$ \limsup_{k \rightarrow \infty}{\P(\hat{S}_{k} \neq \sigma_0)} < \frac{1}{2} \, . $$
\item If $\delta \in \left(\delta_{\sf{maj}},\frac{1}{2}\right)$, and the number of vertices per level satisfies $L_k = o\Big(\!\left(\frac{2}{3(1-2\delta)}\right)^{\! k}\!\Big)$, then reconstruction is impossible in the sense of \eqref{Eq:Strong TV Reconstruction Impossible}:
$$ \lim_{k \rightarrow \infty}{\E\!\left[\left\|P_{\sigma_k|G}^+ - P_{\sigma_k|G}^-\right\|_{\sf{TV}}\right]} = 0 \, . $$
\end{enumerate} 
\end{theorem}

Theorem \ref{Thm:Phase Transition in Random Grid with Majority Rule Processing} is proved in section \ref{Analysis of Majority Rule Processing in Random Grid}. Intuitively, the proof considers the conditional expectation function $\sigma \mapsto \E[\sigma_k|\sigma_{k-1} = \sigma]$ which provides the approximate value of $\sigma_k$ given the value of $\sigma_{k-1}$ for large $k$. This function turns out to have three fixed points when $\delta \in \left(0,\delta_{\sf{maj}}\right)$, and only one fixed point when $\delta \in \left(\delta_{\sf{maj}},\frac{1}{2}\right)$. In the former case, $\sigma_k$ ``moves'' to the largest fixed point when $\sigma_0 = 1$, and to the smallest fixed point when $\sigma_0 = 0$. In the latter case, $\sigma_k$ ``moves'' to the unique fixed point of $\frac{1}{2}$ regardless of the value of $\sigma_0$.\footnote{Note, however, that $\sigma_k \rightarrow \frac{1}{2}$ almost surely as $k \rightarrow 
\infty$ does not imply the impossibility of reconstruction in the sense of \eqref{Eq:TV Reconstruction Impossible}, let alone \eqref{Eq:Strong TV Reconstruction Impossible}. So, a different argument is required to establish such impossibility results.} This provides the guiding intuition for why we can asymptotically decode $\sigma_0$ when $\delta \in \left(0,\delta_{\sf{maj}}\right)$, but not when $\delta \in \left(\delta_{\sf{maj}},\frac{1}{2}\right)$.

It is worth comparing Theorem \ref{Thm:Phase Transition in Random Grid with Majority Rule Processing} with Von Neumann's results in \cite[Section 8]{vonNeumann:56}, where the threshold of $\frac{1}{6}$ is also significant. In \cite[Section 8]{vonNeumann:56}, Von Neumann demonstrates the possibility of reliable computation by constructing a circuit with successive layers of computation and local error correction using $3$-input noisy majority gates. Note that in this model, the gates are noisy (i.e. the gates independently make errors with probability $\delta$), while in our model, the edges (or wires) are noisy.\footnote{See \cite{DobrushinOrtyukov1977} for the relation between gate noise and edge noise.} In his analysis, Von Neumann first derives a simple recursion that captures the effect on the probability of error after applying a single noisy majority gate. Then, he uses a ``heuristic'' fixed point argument to show that as the depth of the circuit grows, the probability of error asymptotically stabilizes at a fixed point value less than $\frac{1}{2}$ if $\delta < \frac{1}{6}$, and the probability of error tends to $\frac{1}{2}$ if $\delta \geq \frac{1}{6}$. Furthermore, he is able to rigorously prove that reliable computation is possible for $\delta < 0.0073$. 

As we mentioned in subsection \ref{Motivation}, Von Neumann's approach to remembering a random initial bit entails using multiple clones of the initial bit as inputs to a noisy circuit with one output, where the output equals the initial bit with probability greater than $\frac{1}{2}$ for ``good'' choices of noisy gates. It is observed in \cite[Section 2]{HajekWeller1991} that a balanced ternary tree circuit, with $k$ layers of $3$-input noisy majority gates and $3^k$ inputs that are all equal to the initial bit, can be used to remember the initial bit. In fact, Von Neumann's heuristic fixed point argument that yields a critical threshold of $\frac{1}{6}$ for reconstruction is accurate in this scenario. Moreover, Hajek and Weller also prove the stronger impossibility result that reliable computation is impossible for formulas (i.e. circuits where the output of each intermediate gate is the input of one other gate) with general $3$-input gates when $\delta \geq \frac{1}{6}$ \cite[Proposition 2]{HajekWeller1991}.

In the brief intuition for our proof of Theorem \ref{Thm:Phase Transition in Random Grid with Majority Rule Processing} given above, the recursion given by the repeated composition of $\sigma \mapsto \E[\sigma_k|\sigma_{k-1} = \sigma]$ seems similar to Von Neumann's recursion in \cite[Section 8]{vonNeumann:56} since both analyze majority gates and yield the same critical threshold of $\frac{1}{6}$. However, our recursion is also quite different in two crucial ways: (a) We prove the $\frac{1}{6}$ threshold for a model where errors occur on the edges rather than at the gates (as mentioned earlier). (b) Since our recursion is defined on the proportion of $1$'s in a layer via conditional expectations, our proof requires exponential concentration inequalities to formalize the intuition provided by the fixed point analysis.
%(c) Our fixed point analysis pertains to the proportion of $1$'s in a layer, rather than the probability of error in Von Neumann's setting.

We now make some pertinent remarks about Theorem \ref{Thm:Phase Transition in Random Grid with Majority Rule Processing}. Firstly, reconstruction is possible in the sense of \eqref{Eq:TV Reconstruction Possible} when $\delta \in \left(0,\delta_{\sf{maj}}\right)$ since the ML decision rule achieves lower probability of error than the majority decision rule,\footnote{It can be seen from monotonicity and symmetry considerations that without knowledge of the random DAG realization $G$, the ML decision rule $f^k_{\sf{ML}}(\sigma_k)$ is equal to the majority decision rule $\hat{S}_k$. On the other hand, with knowledge of the random DAG realization $G$, the ML decision rule $f^k_{\sf{ML}}(\sigma_k,G)$ is not the majority decision rule.
%In fact, simulations illustrate that the distributions $P_{\sigma_k}^+$ and $P_{\sigma_k}^-$ have the \textit{monotone likelihood ratio property}, i.e. the likelihood ratio $\frac{P_{\sigma_k}^+}{P_{\sigma_k}^-}(\sigma)$ is non-decreasing in $\sigma$. This implies that the ML decoder is the majority decoder.
} and reconstruction is impossible in the sense of \eqref{Eq:TV Reconstruction Impossible} when $\delta \in \left(\delta_{\sf{maj}},\frac{1}{2}\right)$ (as explained at the end of subsection \ref{Random Grid Model}). Secondly, in the $\delta \in \left(0,\delta_{\sf{maj}}\right)$ regime, reconstruction is in fact possible under the weaker assumption that $L_k \geq h(\delta) \log(k)$ for some constant $h(\delta)$ (that depends on $\delta$) and all sufficiently large $k$; we will briefly explain this after presenting the proof of Theorem \ref{Thm:Phase Transition in Random Grid with Majority Rule Processing} in section \ref{Analysis of Majority Rule Processing in Random Grid}. Thirdly, the ML decoder $f^k_{\sf{ML}}(\sigma_k)$ is only optimal in the absence of knowledge of the particular graph realization $G$. If the decoder knows the graph $G$, then it can do better and possibly beat the $\delta_{\sf{maj}} = \frac{1}{6}$ threshold. We do note, however, that (except for a vanishing fraction of DAGs) this would require using the full $k$-layer state $X_k$, not just $\sigma_k$ (since the decoder $f^k_{\sf{ML}}(\sigma_k,G)$ does not beat the $\delta_{\sf{maj}}$ threshold on average). Fourthly, the following conjecture is still open: In the random DAG model with $d = 3$ and $L_k = O(\log(k))$, reconstruction is impossible for all choices of Boolean processing functions when $\delta > \delta_{\sf{maj}}$. (A consequence of this conjecture is that majority processing functions are optimal, i.e. they achieve the $\delta_{\sf{maj}}$ reconstruction threshold.) Lastly, it is worth mentioning that for any fixed graph with indegree $d = 3$ and sub-exponential $L_k$, for any choice of Boolean processing functions, and any choice of decoder, it is impossible to reconstruct the root bit when $\delta > \frac{1}{2} - \frac{1}{2\sqrt{3}} = 0.21132...$. This follows from Evans and Schulman's result in \cite{EvansSchulman99}, which we will discuss in subsection \ref{Further Discussion}.
% Can also explain the upper bound on L_k in part 2 intuitively using broadcasting on trees.
% Can also mention that the limsup in part 1 is actually a limit because the majority decision rule is the ML decision rule in this case.

We next present an immediate corollary of Theorem \ref{Thm:Phase Transition in Random Grid with Majority Rule Processing} which states that there exist constant indegree (deterministic) DAGs with $L_k = \omega(\log(k))$ such that reconstruction of the root bit is possible. Formally, we have the following result which is proved in Appendix \ref{Miscellaneous Proofs}.

\begin{corollary}[Existence of DAGs where Reconstruction is Possible]
\label{Cor: Existence of Grids where Reconstruction is Possible}
For any $\delta \in \left(0,\frac{1}{6}\right)$, there exists a DAG $\mathcal{G}$ with $d = 3$ and any $L_k = \omega(\log(k))$ such that if we use majority rules as our Boolean processing functions, then there exists $\epsilon > 0$ such that the probability of error in ML decoding is bounded away from $\frac{1}{2} - \epsilon$:
$$ \forall k \in \N, \enspace \P\!\left(h_{\sf{ML}}^k(X_k,\mathcal{G}) \neq X_{0}\right) \leq \frac{1}{2} - \epsilon $$
where $h_{\sf{ML}}^k(\cdot,\mathcal{G}) : \{0,1\}^{L_k} \rightarrow \{0,1\}$ denotes the ML decision rule at level $k$ of $\mathcal{G}$ based on the full $k$-layer state $X_k$.
\end{corollary}

Until now, we have restricted ourselves to the $d = 3$ case of the random DAG model, because we can always neglect $d-3$ inputs at each processing function if $d > 3$. However, this restriction has only allowed us to prove the existence of DAGs (where reconstruction is possible) when $\delta < \frac{1}{6}$. It can in fact be shown that for every $\delta \in \left(0,\frac{1}{2}\right)$, there exists a DAG with some $d \geq 3$ (that depends on $\delta$), $L_k = \omega(\log(k))$, and all $d$-input majority Boolean processing functions such that reconstruction is possible based on the full $k$-layer state $X_k$.
% Write improvement of this result and Corollary 1 to $L_k = \Omega(\log(k))$.

Our second result considers the setting where the indegree of each node (except the root) is $d = 2$, because it is not immediately obvious that deterministic DAGs (for which reconstruction is possible) exist for $d = 2$. Indeed, it is not entirely clear which Boolean processing functions are good for ``local error correction'' in this scenario. We choose to fix all Boolean functions at even levels of the random DAG model to be the AND rule, and all Boolean functions at odd levels of the model to be the OR rule. We then prove that this random DAG model also exhibits a phase transition phenomenon around a critical threshold of $\delta_{\sf{andor}} \triangleq \frac{3 - \sqrt{7}}{4}$. As before, the next theorem illustrates that for $\delta < \delta_{\sf{andor}}$, the ``biased'' majority decision rule, $\hat{T}_k \triangleq \I\!\left\{\sigma_k \geq t\right\}$ where $t \in (0,1)$ is defined in \eqref{Eq: Middle Fixed Point} in section \ref{Analysis of And-Or Rule Processing in Random Grid}, can asymptotically decode $\sigma_0$, but for $\delta > \delta_{\sf{andor}}$, the ML decision rule cannot asymptotically decode $\sigma_0$. For simplicity, we only analyze this model at even levels.

\begin{theorem}[Phase Transition in Random DAG Model with AND-OR Rule Processing]
\label{Thm:Phase Transition in Random Grid with And-Or Rule Processing}
Let $C(\delta)$ be the constant defined in \eqref{Eq: Lipschitz constant} in section \ref{Analysis of And-Or Rule Processing in Random Grid}. For a random DAG model with $d = 2$, AND processing functions at even levels, and OR processing functions at odd levels, the following phase transition phenomenon occurs around $\delta_{\sf{andor}}$:
\begin{enumerate}
\item If $\delta \in \left(0,\delta_{\sf{andor}}\right)$, and the number of vertices per level satisfies $L_k = \omega(\log(k))$, then reconstruction is possible in the sense that:
$$ \limsup_{k \rightarrow \infty}{\P(\hat{T}_{2k} \neq \sigma_0)} < \frac{1}{2} \, . $$
\item If $\delta \in \left(\delta_{\sf{andor}},\frac{1}{2}\right)$, and the number of vertices per level satisfies $L_k = o\big((C(\delta) + \epsilon)^{-\frac{k}{2}}\big)$ for some $\epsilon \in (0,1-C(\delta))$ (that can depend on $\delta$), then reconstruction is impossible in the sense of \eqref{Eq:Strong TV Reconstruction Impossible}:
$$ \lim_{k \rightarrow \infty}{\E\!\left[\left\|P_{\sigma_{2k}|G}^+ - P_{\sigma_{2k}|G}^-\right\|_{\sf{TV}}\right]} = 0 \, . $$
\end{enumerate} 
\end{theorem} 

Theorem \ref{Thm:Phase Transition in Random Grid with And-Or Rule Processing} is proved in section \ref{Analysis of And-Or Rule Processing in Random Grid}, and many of the remarks pertaining to Theorem \ref{Thm:Phase Transition in Random Grid with Majority Rule Processing} as well as the general intuition for Theorem \ref{Thm:Phase Transition in Random Grid with Majority Rule Processing} also hold for Theorem \ref{Thm:Phase Transition in Random Grid with And-Or Rule Processing}. Furthermore, a corollary analogous to Corollary \ref{Cor: Existence of Grids where Reconstruction is Possible} also holds here.

\subsection{Deterministic 2D Grid Model}
\label{Deterministic 2D Grid Model}

We now turn to deterministic DAG models. As we mentioned earlier, all deterministic DAGs we will analyze will have the structure of a regular 2D grid. A \textit{deterministic 2D grid} consists of a deterministic DAG whose vertices are also Bernoulli random variables and whose edges are independent $\textsf{\small BSC}(\delta)$'s. As with random DAG models, there is a root random variable $X_{0,0} \sim \Ber\!\left(\frac{1}{2}\right)$, and we let $X_k = (X_{k,0},\dots,X_{k,k})$ be the vector of node random variables at distance $k \in \N$ from the root. So, there are $k+1$ nodes at distance $k$. Furthermore, the 2D grid contains the (deterministic) directed edges $(X_{k,j},X_{k+1,j})$ and $(X_{k,j},X_{k+1,j+1})$ for every $k \in \N$ and every $j \in [k+1]$. The underlying graph of such a 2D grid is shown in Figure \ref{Figure: Grid}.

\begin{figure}
\centering
\includegraphics[trim = 50mm 180mm 50mm 30mm, width=0.4\linewidth]{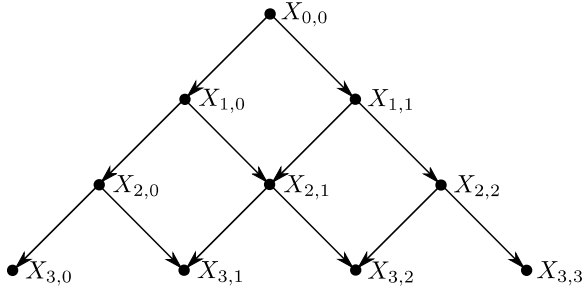} 
\caption{Illustration of a deterministic 2D grid. In this 2D grid, each node is a Bernoulli random variable and each edge is a BSC with parameter $\delta \in \left(0,\frac{1}{2}\right)$. Moreover, each node uses a Boolean processing function to combine its (possibly flipped) input bits.}
\label{Figure: Grid}
\end{figure}

To define the Bayesian network on a deterministic 2D grid, we again fix some crossover probability $\delta \in \left(0,\frac{1}{2}\right)$, and two Boolean functions $f_1:\{0,1\}^2 \rightarrow \{0,1\}$ and $f_{2}:\{0,1\} \rightarrow \{0,1\}$. Then, for any $k \in \N\backslash\!\{0,1\}$ and $j \in \{1,\dots,k-1\}$, we define:\footnote{As mentioned subsection \ref{Random Grid Model}, we can define a more general model where every node $X_{k,j}$ has its own Boolean processing function $f_{k,j}$, but we will only analyze instances of the simpler model presented here.}
\begin{equation}
\label{Eq:det propagation 1}
X_{k,j} = f_{1}(X_{k-1,j-1} \oplus Z_{k,j,1},X_{k-1,j}\oplus Z_{k,j,2})
\end{equation}
and for any $k \in \N\backslash\!\{0\}$, we define:
\begin{equation}
\label{Eq:det propagation 2}
X_{k,0} = f_{2}(X_{k-1,0} \oplus Z_{k,0,2}) \enspace \text{and} \enspace X_{k,k} = f_{2}(X_{k-1,k-1} \oplus Z_{k,k,1})
\end{equation}
where $\{Z_{k,j,i} : k \in \N\backslash\!\{0\}, j \in [k+1], i \in \{1,2\}\}$ are i.i.d $\Ber(\delta)$ random variables that are independent of everything else. Together, \eqref{Eq:det propagation 1} and \eqref{Eq:det propagation 2} characterize the conditional distribution of any $X_{k,j}$ given its parents.

As before, we have a Markov chain $\{X_k : k \in \N\}$, and our goal is to determine whether or not the value at the root $X_{0}$ can be decoded from the observations $X_k$ as $k \rightarrow \infty$. In all the cases that we will consider in this paper, we will prove that reconstruction is impossible in the sense that: 
\begin{equation}
\label{Eq: Deterministic Impossibility of Reconstruction}
\lim_{k \rightarrow \infty}{\left\|P_{X_k}^+ - P_{X_k}^-\right\|_{\sf{TV}}} = 0 
\end{equation}
where $P_{X_k}^+$ and $P_{X_k}^-$ are the conditional distributions of $X_k$ given $X_0 = 1$ and $X_0 = 0$, respectively. The condition in \eqref{Eq: Deterministic Impossibility of Reconstruction} is of course equivalent to the ML decision rule failing to decode $X_0$ from $X_k$ as $k \rightarrow \infty$.

\subsection{Results on Deterministic 2D Grids}

Deterministic 2D grids are much harder to analyze than random DAG models due to the dependence between adjacent nodes in a given layer. As mentioned earlier, we analyze the setting where all Boolean processing functions in the 2D grid with two inputs are the same, and all Boolean processing functions in the 2D grid with one input are the \textit{identity} rule. Our first result shows that reconstruction is impossible for all $\delta \in \left(0,\frac{1}{2}\right)$ when AND processing functions are used.

\begin{theorem}[Deterministic AND 2D Grid]
\label{Thm: Deterministic And Grid}
If $\delta \in \left(0,\frac{1}{2}\right)$, and all Boolean processing functions with two inputs in the deterministic 2D grid are the AND rule, then reconstruction is impossible in the sense of \eqref{Eq: Deterministic Impossibility of Reconstruction}, i.e. $\lim_{k \rightarrow \infty}{\left\|P_{X_k}^+ - P_{X_k}^-\right\|_{\sf{TV}}} = 0$.
\end{theorem} 

Theorem \ref{Thm: Deterministic And Grid} is proved in section \ref{Analysis of Deterministic And Grid}. The proof couples the 2D grid starting at $X_{0,0} = 0$ with the 2D grid starting at $X_{0,0} = 1$, and ``runs'' them together. Using a phase transition result concerning \textit{bond percolation} on 2D grids, we show that we eventually reach a layer where the values of all nodes in the first grid equal the values of the corresponding nodes in the second grid. So, the two 2D grids ``couple'' almost surely regardless of their starting state. This implies that we cannot decode the starting state by looking at nodes in layer $k$ as $k \rightarrow \infty$. We note that in order to prove that the two 2D grids ``couple,'' we have to consider two different regimes of $\delta$ and provide separate arguments for each. The details of these arguments are presented in section \ref{Analysis of Deterministic And Grid}. 

Our second result shows that reconstruction is impossible for all $\delta \in \left(0,\frac{1}{2}\right)$ when XOR processing functions are used.

\begin{theorem}[Deterministic XOR 2D Grid]
\label{Thm: Deterministic Xor Grid}
If $\delta \in \left(0,\frac{1}{2}\right)$, and all Boolean processing functions with two inputs in the deterministic 2D grid are the XOR rule, then reconstruction is impossible in the sense of \eqref{Eq: Deterministic Impossibility of Reconstruction}, i.e. $\lim_{k \rightarrow \infty}{\left\|P_{X_k}^+ - P_{X_k}^-\right\|_{\sf{TV}}} = 0$.
\end{theorem}

Theorem \ref{Thm: Deterministic Xor Grid} is proved in section \ref{Analysis of Deterministic Xor Grid}. In the XOR 2D grid, every node at level $k$ can be written as a (binary) linear combination of the root bit and all the BSC noise random variables in the grid up to level $k$. This linear relationship can be captured by a binary matrix. The main idea of the proof is to perceive this matrix as a parity check matrix of a linear code. The problem of inferring $X_{0,0}$ from $X_k$ turns out to be equivalent to decoding the first bit of a codeword drawn uniformly from this code after observing a noisy version of the codeword. Basic facts from coding theory can then be used to complete the proof.

We remark that Theorems \ref{Thm: Deterministic And Grid} and \ref{Thm: Deterministic Xor Grid} seem intuitively obvious from the random DAG model perspective. For example, consider the random DAG model with $d = 2$, $L_k = k+1$, and all AND processing functions. Then, the conditional expectation function $\sigma \mapsto \E[\sigma_k|\sigma_{k-1} = \sigma]$ has only one fixed point regardless of the value of $\delta \in \left(0,\frac{1}{2}\right)$, and we intuitively expect $\sigma_k$ to tend to this fixed point (which roughly captures the equilibrium between AND gates killing $1$'s and $\textsf{\small BSC}(\delta)$'s producing new $1$'s) as $k \rightarrow \infty$. So, reconstruction is impossible in this random DAG model, which suggests that reconstruction is also impossible in the deterministic 2D AND grid. Although Theorems \ref{Thm: Deterministic And Grid} and \ref{Thm: Deterministic Xor Grid} are intuitively easy to understand in this way, we emphasize that they are nontrivial to prove--see sections \ref{Analysis of Deterministic And Grid} and \ref{Analysis of Deterministic Xor Grid}. 

The impossibility of reconstruction in Theorems \ref{Thm: Deterministic And Grid} and \ref{Thm: Deterministic Xor Grid} also seems intuitively plausible due to the ergodicity results for numerous 1D (probabilistic) cellular automata--see e.g. \cite{Gray82} and the references therein. However, there are two key differences between deterministic 2D grids and 1D cellular automata. Firstly, the main question in the study of 1D cellular automata is whether a given automaton is ergodic, i.e. whether the Markov process defined by it converges to a unique invariant probability measure on the configuration space for all initial configurations. This question of ergodicity is typically addressed by considering the convergence of finite-dimensional distributions over the sites (i.e. weak convergence). Hence, for many 1D cellular automata that have special characteristics (such as translation invariance, finite range, positivity, and attractiveness/monotonicity, cf. \cite{Gray82}), it suffices to consider the convergence of distributions on finite intervals (e.g. marginal distribution at a given site). In contrast to this setting, we are concerned with the stronger notion of convergence in TV distance. Indeed, Theorems \ref{Thm: Deterministic And Grid} and \ref{Thm: Deterministic Xor Grid} show that the TV distance between $P_{X_k}^+$ and $P_{X_k}^{-}$ vanishes as $k \rightarrow \infty$.

Secondly, since a 1D cellular automaton has infinitely many sites, the problem of remembering a bit in a cellular automaton corresponds to distinguishing between the ``all zeros'' and ``all ones'' initial configurations. On the other hand, a deterministic 2D grid can be construed as a 1D cellular automaton with boundary conditions; each level $k$ corresponds to an instance in (discrete) time, and there are $L_k$ sites at time $k$. Moreover, its initial configuration has only one copy of the initial bit as opposed to infinitely many copies. As a result, compared a deterministic 2D grid, a 1D cellular automaton (without boundary conditions) intuitively appears to have a stronger separation between the two initial states as time progresses. The aforementioned boundary conditions form another barrier to translating results from the 1D cellular automata literature to deterministic 2D grids.

It is also worth mentioning that most results on 1D cellular automata pertain to the continuous time setting--see e.g. \cite{Liggett78, Gray82} and the references therein. This is because sites are updated one by one in a continuous time automaton, but they are updated in parallel in a discrete time automaton. So, the discrete time setting is often harder to analyze. (One of the only known discrete time 1D cellular automaton ergodicity results, for the $3$-input majority vote model, is outlined in \cite[Section 3]{Gray87}.) This is another reason why results from the 1D cellular automata literature cannot be easily used for our model.

\subsection{Further Discussion on Impossibility Results}
\label{Further Discussion}

In this subsection, we present and discuss two impossibility results pertaining to both deterministic and random DAG models (where the former correspond to Bayesian networks on specific realizations of $G$ as defined in subsection \ref{Random Grid Model}). The first result illustrates that if $L_k \leq \log(k)/(d \log(1/(2\delta)))$ for every sufficiently large $k$ (i.e. $L_k$ grows very ``slowly''), then reconstruction is impossible regardless of the choice of Boolean processing functions and the choice of decision rule.

\begin{proposition}[Slow Growth of Layers]
\label{Prop: Slow Growth of Layers}
Suppose that the number of vertices per level satisfies:
$$ \exists K \in \N, \forall k \geq K, \enspace L_k \leq \frac{\log(k)}{d \log\!\left(\frac{1}{2\delta}\right)} \, . $$
Then, reconstruction is impossible and we have:
\begin{enumerate}
\item for a deterministic DAG:
$$ \lim_{k \rightarrow \infty}{\left\|P_{X_k}^+ - P_{X_k}^-\right\|_{\sf{TV}}} = 0 \, . $$
\item for a random DAG:
$$ \lim_{k \rightarrow \infty}{\E\!\left[\left\|P_{X_{k}|G}^+ - P_{X_{k}|G}^-\right\|_{\sf{TV}}\right]} = 0 \, . $$
\end{enumerate}
\end{proposition} 

This proposition is proved in Appendix \ref{Miscellaneous Proofs}. In part 2 of Proposition \ref{Prop: Slow Growth of Layers}, $P_{X_k|G}^+$ and $P_{X_k|G}^-$ denote the conditional distributions of $X_k$ given $\{X_0 = 1,G\}$ and $\{X_0 = 0,G\}$ respectively, which shows that reconstruction is impossible for random DAGs even if the particular DAG realization $G$ is known and the decoder can access the entire $k$-layer state $X_k$. Furthermore, part 2 also clearly implies that reconstruction is impossible in the sense of \eqref{Eq:Strong TV Reconstruction Impossible}. Therefore, Proposition \ref{Prop: Slow Growth of Layers} illustrates that our assumption that $L_k = \omega(\log(k))$ (or as we mentioned earlier, the weaker assumption that $L_k \geq h(\delta) \log(k)$ for some constant $h(\delta)$ and all sufficiently large $k$) for reconstruction to be possible in Theorems \ref{Thm:Phase Transition in Random Grid with Majority Rule Processing} and \ref{Thm:Phase Transition in Random Grid with And-Or Rule Processing} is in fact necessary. In contrast, consider a deterministic DAG with no restrictions (i.e. no bounded indegree assumption) except for the size of $L_k$. Then, each node at level $k$ is connected to all $L_{k-1}$ nodes at level $k-1$. In this scenario, the proof technique of part 1 of Proposition \ref{Prop: Slow Growth of Layers} in Appendix \ref{Miscellaneous Proofs} can be used to show that reconstruction is impossible when $L_k \leq \sqrt{\log(k)/\log(1/(2\delta))}$ for all sufficiently large $k$. Moreover, this scaling of $L_k = O\big(\sqrt{\log(k)}\big)$ is tight, because if we let every Boolean processing function be the majority vote, then the proof of part 1 of Theorem \ref{Thm:Phase Transition in Random Grid with Majority Rule Processing} in section \ref{Analysis of Majority Rule Processing in Random Grid} can be executed mutatis mutandis to show that reconstruction is possible for every $\delta \in \left(0,\frac{1}{2}\right)$ when $L_k = \Omega\big(\sqrt{\log(k)}\big)$ and $\limsup_{k \rightarrow \infty}{L_k/L_{k-1}} < +\infty$.

The second impossibility result we present is an important result from the noisy circuits literature due to Evans and Schulman \cite{EvansSchulman99}. Evans and Schulman studied Von Neumann's noisy computation model (which we briefly discussed in subsection \ref{Results on Random DAG Models}), and established general conditions under which reconstruction is impossible in deterministic DAGs due to the decay of mutual information between $X_0$ and $X_k$. Recall that for two discrete random variables $X \in \X$ and $Y \in \Y$ (where $|\X|,|\Y| < \infty$), with joint probability mass function $P_{X,Y}$ and marginals $P_X$ and $P_Y$ respectively, the \textit{mutual information} (in bits) between them is defined as:
\begin{equation}
I(X;Y) \triangleq \sum_{x \in \X}\sum_{y \in \Y}{P_{X,Y}(x,y) \log_2\!\left(\frac{P_{X,Y}(x,y)}{P_X(x)P_Y(y)}\right)} 
\end{equation}
where $\log_2(\cdot)$ is the binary logarithm, and we assume that $0 \log_2\!\left(\frac{0}{q}\right) = 0$ for any $q \geq 0$, and $p \log_2\!\left(\frac{p}{0}\right) = \infty$ for any $p > 0$ (due to continuity considerations). We present a specialization of \cite[Lemma 2]{EvansSchulman99} for our setting as Proposition \ref{Prop: Evans Schulman} below. This proposition portrays that if $L_k$ is sub-exponential and the parameters $\delta$ and $d$ satisfy $(1 - 2\delta)^2 d < 1$, then reconstruction is impossible in deterministic DAGs regardless of the choice of Boolean processing functions and the choice of decision rule.

\begin{proposition}[Decay of Mutual Information {\cite[Lemma 2]{EvansSchulman99}}]
\label{Prop: Evans Schulman}
For any deterministic DAG model, we have:
$$ I(X_0;X_k) \leq L_k \left((1-2\delta)^2 d\right)^k $$
where $L_k d^k$ is the total number of paths from $X_0$ to layer $X_k$, and $(1-2\delta)^{2k}$ can be construed as the overall contraction of mutual information along each path. Therefore, if $(1 - 2\delta)^2 d < 1$ and $L_k = o\!\left(1/((1-2\delta)^2 d)^k\right)$, then $\lim_{k \rightarrow \infty}{I(X_0;X_k)} = 0$, which implies that $\lim_{k \rightarrow \infty}{\left\|P_{X_k}^+ - P_{X_k}^-\right\|_{\sf{TV}}} = 0$. 
\end{proposition}

We make some pertinent remarks about this result. Firstly, Evans and Schulman's original analysis assumes that gates are noisy as opposed to edges (in accordance with Von Neumann's setup), but the re-derivation of \cite[Lemma 2]{EvansSchulman99} in \cite[Corollary 7]{GraphSDPI} illustrates that the result also holds for our model. In fact, the \textit{site percolation} analysis in \cite[Section 3]{GraphSDPI} (which we will briefly delineate later) improves upon Evans and Schulman's estimate. Furthermore, this analysis illustrates that the bound in Proposition \ref{Prop: Evans Schulman} also holds for all choices of random Boolean processing functions.

Secondly, while Proposition \ref{Prop: Evans Schulman} holds for deterministic DAGs, we can easily extend it for random DAG models. Indeed, the random DAG model inherits the inequality in Proposition \ref{Prop: Evans Schulman} pointwise:
\begin{equation}
I(X_0;X_k|G = \mathcal{G}) \leq L_k \left((1-2\delta)^2 d\right)^k
\end{equation}
for every realization of the random DAG $G = \mathcal{G}$, where $I(X_0;X_k|G = \mathcal{G})$ is the mutual information between $X_0$ and $X_k$ computed using the joint distribution of $X_0$ and $X_k$ given $G = \mathcal{G}$. Taking expectations with respect to $G$, we get:
\begin{equation}
I(\sigma_0;\sigma_k) = I(X_0;X_k) \leq I(X_0;X_k|G) \leq L_k \left((1-2\delta)^2 d\right)^k
\end{equation}
where $I(X_0;X_k|G)$ is the conditional mutual information (i.e. the expected value of $I(X_0;X_k|G = \mathcal{G})$ with respect to $G$), the equality holds because $\sigma_k$ is a sufficient statistic of $X_k$ for performing inference about $\sigma_0$ (cf. \cite[Section 3.1]{InfoTheoryNotes}), and the first inequality follows from the chain rule for mutual information and the fact that $X_0$ is independent of $G$. Hence, if $L_k$ is sub-exponential and $(1 - 2\delta)^2 d < 1$, then reconstruction is impossible in the sense of \eqref{Eq:TV Reconstruction Impossible} in the random DAG model regardless of the choice of Boolean processing functions and the choice of decision rule. (It is straightforward to see from the previous discussion that reconstruction is also impossible in the sense of \eqref{Eq:Strong TV Reconstruction Impossible}.)

Thirdly, Evans and Schulman's result in Proposition \ref{Prop: Evans Schulman} provides an upper bound on the critical threshold of $\delta$ above which reconstruction of the root bit is impossible. Indeed, the condition, $(1 - 2\delta)^2 d < 1$, under which mutual information decays can be rewritten as (cf. the discussion in \cite[p.2373]{EvansSchulman99}):
\begin{equation}
\delta_{\sf{ES}}(d) \triangleq \frac{1}{2} - \frac{1}{2\sqrt{d}} < \delta < \frac{1}{2}
\end{equation}
and reconstruction is impossible for deterministic or random DAGs in this regime of $\delta$ provided $L_k$ is sub-exponential. As a sanity check, we can verify that $\delta_{\sf{ES}}(2) = 0.14644... > 0.08856... = \delta_{\sf{andor}}$ in the context of Theorem \ref{Thm:Phase Transition in Random Grid with And-Or Rule Processing}, and $\delta_{\sf{ES}}(3) = 0.21132... > 0.16666... = \delta_{\sf{maj}}$ in the context of Theorem \ref{Thm:Phase Transition in Random Grid with Majority Rule Processing} (as mentioned in subsection \ref{Results on Random DAG Models}). Although $\delta_{\sf{ES}}(d)$ is a general upper bound on the critical threshold for reconstruction, in this paper, it is not particularly useful because we analyze explicit processing functions and decision rules, and derive specific bounds that characterize the corresponding thresholds.

Fourthly, it is worth comparing $\delta_{\sf{ES}}(d)$ (which comes from a site percolation argument, cf. \cite[Section 3]{GraphSDPI}) to an upper bound on the critical threshold for reconstruction derived from bond percolation. To this end, consider the random DAG model, and recall that the $\textsf{\small BSC}(\delta)$'s along each edge generate independent bits with probability $2\delta$ (as shown in the proof of Proposition \ref{Prop: Slow Growth of Layers} in Appendix \ref{Miscellaneous Proofs}). So, we can perform bond percolation so that each edge is independently ``removed'' with probability $2\delta$. It can be shown by analyzing this bond percolation process that reconstruction is impossible when $\frac{1}{2} - \frac{1}{2d} < \delta < \frac{1}{2}$. Therefore, the Evans-Schulman upper bound of $\delta_{\sf{ES}}(d)$ is tighter than the bond percolation upper bound: $\delta_{\sf{ES}}(d) < \frac{1}{2} - \frac{1}{2d}$.

Finally, we briefly delineate how the site percolation approach in \cite[Section 3]{GraphSDPI} allows us to prove that reconstruction is impossible in the random DAG model for the $(1-2\delta)^2 d = 1$ case as well. Consider a site percolation process where each node $X_{k,j}$ (for $k \in \N\backslash\!\{0\}$ and $j \in [L_k]$) is independently ``open'' with probability $(1-2\delta)^2$, and ``closed'' with probability $1-(1-2\delta)^2$. (Note that $X_{0,0}$ is open almost surely.) For every $k \in \N\backslash\!\{0\}$, let $p_k$ denote the probability that there is an ``open connected path'' from $X_{0}$ to $X_{k}$ (i.e. there exist $j_1 \in [L_1],\dots,j_k \in [L_k]$ such that $(X_{0,0},X_{1,j_1}),(X_{1,j_1},X_{2,j_2}),\dots,(X_{k-1,j_{k-1}},X_{k,j_k})$ are directed edges in the random DAG $G$ and $X_{1,j_1},\dots,X_{k,j_k}$ are all open). It can be deduced from \cite[Theorem 5]{GraphSDPI} that for any $k \in \N\backslash\!\{0\}$: 
\begin{equation}
\label{Eq: Polyanskiy-Wu Bound}
I(X_0;X_k|G) \leq p_k \, . 
\end{equation} 
Next, for each $k \in \N$, define the random variable:
\begin{equation}
\lambda_k \triangleq \frac{1}{L_k} \sum_{j \in [L_k]}{\I\!\left\{X_{k,j} \text{ is open and connected}\right\}} 
\end{equation}
which is the proportion of open nodes at level $k$ that are connected to the root by an open path. (Note that $\lambda_0 = 1$.) It is straightforward to verify (using Bernoulli's inequality) that for any $k \in \N\backslash\!\{0\}$:
\begin{equation}
\label{Eq: Recursion}
\E\!\left[\lambda_k|\lambda_{k-1}\right] = (1-2\delta)^2 \!\left(1-(1-\lambda_{k-1})^d\right) \leq (1-2\delta)^2 d \lambda_{k-1} \, . 
\end{equation}
Observe that by Markov's inequality and the recursion from \eqref{Eq: Recursion}, $\E\!\left[\lambda_k\right] \leq (1-2\delta)^2 d \, \E\!\left[\lambda_{k-1}\right]$, we have:
\begin{equation}
\label{Eq: Markov argument}
p_k = \P\!\left(\lambda_k \geq \frac{1}{L_k}\right) \leq L_k \E\!\left[\lambda_k\right] \leq L_k \left((1-2\delta)^2 d\right)^k
\end{equation}
which recovers Evans and Schulman's result (Proposition \ref{Prop: Evans Schulman}) in the context of the random DAG model. Indeed, if $(1-2\delta)^2 d < 1$ and $L_k = o\!\left(1/((1-2\delta)^2 d)^k\right)$, then $\lim_{k \rightarrow \infty}{p_k} = 0$, and as a result, $\lim_{k \rightarrow \infty}{I(X_0;X_k|G)} = 0$ by \eqref{Eq: Polyanskiy-Wu Bound}. On the other hand, when $(1-2\delta)^2 d = 1$, taking expectations and applying Jensen's inequality to the equality in \eqref{Eq: Recursion} produces:
\begin{equation}
\E\!\left[\lambda_k\right] \leq (1-2\delta)^2 \!\left(1-(1-\E\!\left[\lambda_{k-1}\right])^d\right) . 
\end{equation}
This implies that $\E\!\left[\lambda_k\right] \leq F^{-1}(k)$ for every $k \in \N$ using the estimate in \cite[Appendix A]{ContractionCoefficients}, where $F:[0,1] \rightarrow \R_{+}, \, F(t) = \int_{t}^{1}{\frac{1}{f(\tau)} \, d\tau}$ with $f:[0,1] \rightarrow [0,1], \, f(t) = t - (1-2\delta)^2 \! \left(1-(1-t)^d\right)$, and $F^{-1}:\R_{+} \rightarrow [0,1]$ is well-defined. Since $f(t) \geq \frac{d-1}{2} t^2$ for all $t \in [0,1]$, it is straightforward to show that:
\begin{equation}
\E\!\left[\lambda_k\right] \leq F^{-1}(k) \leq \frac{2}{(d-1) k} \, . 
\end{equation}
Therefore, the Markov's inequality argument in \eqref{Eq: Markov argument} illustrates that if $(1-2\delta)^2 d = 1$ and $L_k = o(k)$, then $\lim_{k \rightarrow \infty}{p_k} = 0$ and reconstruction is impossible in the random DAG model due to \eqref{Eq: Polyanskiy-Wu Bound}. Furthermore, the condition on $L_k$ can be improved to $L_k = O(k \log(k))$ using a more sophisticated Borel-Cantelli type of argument.

\section{Analysis of Majority Rule Processing in Random DAG Model}
\label{Analysis of Majority Rule Processing in Random Grid}

In this section, we prove Theorem \ref{Thm:Phase Transition in Random Grid with Majority Rule Processing}. To this end, we first make some pertinent observations. Recall that we have a random DAG model with $d = 3$, and all Boolean functions are the majority rule, i.e. $f_{k}(x_1,x_2,x_3) = \maj(x_1,x_2,x_3)$ for every $k \in \N\backslash\!\{0\}$. Suppose we are given that $\sigma_{k-1} = \sigma$. Then, $X_{k,j} = \maj(\Ber(\sigma \star \delta),\Ber(\sigma \star \delta),\Ber(\sigma \star \delta))$ for three i.i.d. Bernoulli random variables for every $k \in \N\backslash\!\{0\}$ and every $j \in [L_{k}]$. Since we have:
\begin{align}
\P(X_{k,j} = 1|\sigma_{k-1} = \sigma) & = (\sigma \star \delta)^3 + 3 (\sigma \star \delta)^2 (1 - \sigma \star \delta) \\
& = -2(1-2\delta)^3 \sigma^3 + 3(1-2\delta)^3 \sigma^2 + 6\delta(1-\delta)(1-2\delta) \sigma + \delta^2 (3-2\delta) \\
\label{Eq:Conditional Expectation}
& = \E[\sigma_k|\sigma_{k-1} = \sigma] \, ,
\end{align}
$X_{k,j}$ are i.i.d. $\Ber(g(\sigma))$ for $j \in [L_{k}]$, and $L_k \sigma_k \sim \textsf{\small binomial}(L_k,g(\sigma))$, where we define $g:[0,1] \rightarrow [0,1]$ as:
\begin{equation}
 g(\sigma) \triangleq (\sigma \star \delta)^3 + 3 (\sigma \star \delta)^2 (1 - \sigma \star \delta)
\end{equation}
and its derivative $g^{\prime}:[0,1] \rightarrow \R^{+}$ is:
\begin{equation}
 g^{\prime}(\sigma) = 6(1-2\delta)(\sigma \star \delta)(1 - \sigma \star \delta) \geq 0 \, .
\end{equation}
This is a quadratic function of $\sigma$ with maximum value $\max_{\sigma \in [0,1]}{g^{\prime}(\sigma)} = \frac{3}{2}(1-2\delta)$ achieved at $\sigma = \frac{1}{2}$. Hence, $\frac{3}{2}(1-2\delta)$ is the Lipschitz constant of $g$ over $[0,1]$. 

There are two regimes of interest when we consider the contraction properties and fixed point structure of $g$. In the $\delta \in \left(0,\frac{1}{6}\right)$ regime, the Lipschitz constant $\frac{3}{2}(1-2\delta) \in \left(1,\frac{3}{2}\right)$ is greater than $1$. Furthermore, to compute the fixed points of $g$, notice that:
\begin{align}
g(\sigma) - \sigma & = -2(1-2\delta)^3 \sigma^3 + 3(1-2\delta)^3 \sigma^2 + (6\delta(1-\delta)(1-2\delta) - 1) \sigma + \delta^2 (3-2\delta) \nonumber \\
& = \left(\sigma - \frac{1}{2}\right)\!(-2(1-2\delta)^3 \sigma^2 + 2(1-2\delta)^3 \sigma - 2 \delta^2 (3-2\delta)) 
\end{align}
which means that $g$ has three fixed points (or the roots of $g(\sigma) - \sigma$):
\begin{equation}
\sigma = \frac{1}{2}, \frac{1}{2}\left(1 + \sqrt{\frac{1-6\delta}{(1-2\delta)^3}}\right), \frac{1}{2}\left(1 - \sqrt{\frac{1-6\delta}{(1-2\delta)^3}}\right)
\end{equation}
using the quadratic formula. When $\delta \in \left(0,\frac{1}{6}\right)$, let us define the largest fixed point of $g$ as:
\begin{equation}
\hat{\sigma} \triangleq \frac{1}{2}\left(1 + \sqrt{\frac{1-6\delta}{(1-2\delta)^3}}\right)
\end{equation}
so that $g$ has the fixed points $\sigma = 1-\hat{\sigma},\frac{1}{2},\hat{\sigma}$. In contrast, in the $\delta \in \left(\frac{1}{6},\frac{1}{2}\right)$ regime, the Lipschitz constant $\frac{3}{2}(1-2\delta) \in (0,1)$ is less than $1$, and the only fixed point of $g$ is $\sigma = \frac{1}{2}$. (We also mention that when $\delta = \frac{1}{6}$, the Lipschitz constant $\frac{3}{2}(1-2\delta) = 1$, and $g$ has one fixed point at $\sigma = \frac{1}{2}$.) We now prove Theorem \ref{Thm:Phase Transition in Random Grid with Majority Rule Processing}.

\renewcommand{\proofname}{\bfseries \emph{Proof of Theorem \ref{Thm:Phase Transition in Random Grid with Majority Rule Processing}}}

\begin{proof} 
We first prove that $\delta \in \left(0,\frac{1}{6}\right)$ implies $\limsup_{k \rightarrow \infty}{\P(\hat{S}_{k} \neq \sigma_0)} < \frac{1}{2}$. To establish this, we begin by defining a useful ``monotone Markovian coupling'' (see \cite[Chapter 5]{MarkovMixing} for basic definitions of Markovian couplings). Let $\{\sigma^+_k : k \in \N\}$ and $\{\sigma^-_k : k \in \N\}$ denote versions of the Markov chain $\{\sigma_k : k \in \N\}$ (i.e. with the same transition kernels) initialized at $\sigma^+_0 = 1$ and $\sigma^-_0 = 0$, respectively. In particular, the marginal distributions of $\sigma_k^+$ and $\sigma_k^-$ are $P^+_{\sigma_k}$ and $P^-_{\sigma_k}$, respectively. We construct the monotone Markovian coupling $\{(\sigma^+_k,\sigma^-_k) : k \in \N\}$ between the Markov chains $\{\sigma^+_k : k \in \N\}$ and $\{\sigma^-_k : k \in \N\}$ such that for every $k \in \N$, $\sigma^+_k \geq \sigma^-_k$ almost surely. Notice that $1 = \sigma^+_0 \geq \sigma^-_0 = 0$ is true by assumption. Suppose for some $k \in \N$, $\sigma^+_k \geq \sigma^-_k$ almost surely. We define the conditional distribution of $(\sigma^+_{k+1},\sigma^-_{k+1})$ given $(\sigma^+_k,\sigma^-_k) = (\sigma^+,\sigma^-)$ as the well-known monotone coupling of $L_{k+1} \sigma^+_{k+1} \sim \textsf{\small binomial}(L_{k+1},g(\sigma^+))$ and $L_{k+1} \sigma^-_{k+1} \sim \textsf{\small binomial}(L_{k+1},g(\sigma^-))$ so that $\sigma^+_{k+1} \geq \sigma^-_{k+1}$ almost surely given $(\sigma^+_k,\sigma^-_k) = (\sigma^+,\sigma^-)$. Such a monotone coupling exists because $g(\sigma^+) \geq g(\sigma^-)$ since $g$ is a non-decreasing function and $\sigma^+ \geq \sigma^-$. This recursively generates a Markov chain $\{(\sigma^+_k,\sigma^-_k) : k \in \N\}$ with the following properties:
\begin{itemize}
\item[1)] The ``marginal'' Markov chains are $\{\sigma^+_k : k \in \N\}$ and $\{\sigma^-_k : k \in \N\}$.
\item[2)] For every $j > k \geq 1$, $\sigma_{j}^+$ is conditionally independent of $\sigma^-_0,\dots,\sigma^-_k,\sigma^+_0,\dots,\sigma^+_{k-1}$ given $\sigma^+_k$, and $\sigma_{j}^-$ is conditionally independent of $\sigma^+_0,\dots,\sigma^+_k,\sigma^-_0,\dots,\sigma^-_{k-1}$ given $\sigma^-_k$. 
\item[3)] For every $k \in \N$, $\sigma^+_k \geq \sigma^-_k$ almost surely.
\end{itemize}
In the sequel, probabilities of events that depend on the random variables in $\{(\sigma^+_k,\sigma^-_k) : k \in \N\}$ are defined with respect to this Markovian coupling. We next prove that there exists $\epsilon > 0$ such that:
\begin{equation}
\label{Eq: Stability whp}
\forall k \in \N\backslash\!\{0\}, \enspace \P\!\left(\left.\sigma_{k}^+ \geq \hat{\sigma} - \epsilon \, \right| \sigma_{k-1}^+ \geq \hat{\sigma} - \epsilon,A_{k,j}\right) \geq 1 - \exp\!\left(-2 L_k \gamma(\epsilon)^2 \right) 
\end{equation}
where $\gamma(\epsilon) \triangleq g(\hat{\sigma} - \epsilon) - (\hat{\sigma} - \epsilon) > 0$, and $A_{k,j}$ with $0 \leq j < k$ is the non-zero probability event defined as: 
$$ A_{k,j} \triangleq \left\{ 
\begin{array}{lcl}
\{\sigma_{j}^- \leq 1-\hat{\sigma} + \epsilon\} & , & 0 \leq j = k-1 \\
\{\sigma_{k-2}^+ \geq \hat{\sigma} - \epsilon,\dots,\sigma_{j}^+ \geq \hat{\sigma} - \epsilon\} \cap \{\sigma_{j}^- \leq 1-\hat{\sigma} + \epsilon\} & , & 0 \leq j \leq k-2
\end{array} \right. \, .
$$ 
Since $g^{\prime}(\hat{\sigma}) < 1$ and $g(\hat{\sigma}) = \hat{\sigma}$, $g(\hat{\sigma} - \epsilon) > \hat{\sigma} - \epsilon$ for sufficiently small $\epsilon > 0$. Fix any such $\epsilon > 0$ such that $\gamma(\epsilon) > 0$. Recall that $L_k \sigma_k \sim \textsf{\small binomial}(L_k,g(\sigma))$ given $\sigma_{k-1} = \sigma$. This implies that for every $k \in \N\backslash\!\{0\}$ and every $0 \leq j < k$:
$$ \P\!\left(\left.\sigma_k^+ < g\!\left(\sigma_{k-1}^+\right) - \gamma(\epsilon) \, \right|\sigma_{k-1}^+ = \sigma,A_{k,j}\right) = \P(\sigma_k < g(\sigma_{k-1}) - \gamma(\epsilon)|\sigma_{k-1} = \sigma) \leq \exp\!\left(-2 L_k \gamma(\epsilon)^2 \right) $$
where the equality follows from property 2 of our Markovian coupling, and the inequality follows from \eqref{Eq:Conditional Expectation} and Hoeffding's inequality \cite[Theorem 1]{HoeffdingInequality}. As a result, we have:
\begin{align*}
\sum_{\sigma \geq \hat{\sigma} - \epsilon} \P\!\left(\left.\sigma_k^+ < g\!\left(\sigma_{k-1}^+\right) - \gamma(\epsilon) \, \right|\sigma_{k-1}^+ = \sigma,A_{k,j}\right) & \, \P\!\left(\left.\sigma_{k-1}^+ = \sigma \, \right| A_{k,j}\right) \\
& \leq \exp\!\left(-2 L_k \gamma(\epsilon)^2 \right) \sum_{\sigma \geq \hat{\sigma} - \epsilon}{\P\!\left(\left.\sigma_{k-1}^+ = \sigma \, \right| A_{k,j}\right)} \\
\P\!\left(\left.\sigma_k^+ < g\!\left(\sigma_{k-1}^+\right) - \gamma(\epsilon), \sigma_{k-1}^+ \geq \hat{\sigma} - \epsilon \, \right|A_{k,j}\right) & \leq \exp\!\left(-2 L_k \gamma(\epsilon)^2 \right) \P\!\left(\left.\sigma_{k-1}^+ \geq \hat{\sigma} - \epsilon \, \right| A_{k,j}\right) \\
\P\!\left(\left.\sigma_k^+ < g\!\left(\sigma_{k-1}^+\right) - \gamma(\epsilon) \, \right| \sigma_{k-1}^+ \geq \hat{\sigma} - \epsilon, A_{k,j}\right) & \leq \exp\!\left(-2 L_k \gamma(\epsilon)^2 \right) \, .
\end{align*}
Finally, notice that $\sigma_k^+ < \hat{\sigma} - \epsilon = g(\hat{\sigma} - \epsilon) - \gamma(\epsilon)$ implies that $\sigma_k^+ < g(\sigma_{k-1}^+) - \gamma(\epsilon)$ when $\sigma_{k-1}^+ \geq \hat{\sigma} - \epsilon$ (since $g$ is non-decreasing and $g(\sigma_{k-1}^+) \geq g(\hat{\sigma} - \epsilon)$). This produces: 
$$ \P\!\left(\left.\sigma_k^+ < \hat{\sigma} - \epsilon \, \right| \sigma_{k-1}^+ \geq \hat{\sigma} - \epsilon, A_{k,j}\right) \leq \exp\!\left(-2 L_k \gamma(\epsilon)^2 \right) $$
which in turn establishes \eqref{Eq: Stability whp}.

Now fix any $\tau > 0$, and choose a sufficiently large value $K = K(\epsilon,\tau) \in \N$ (that depends on $\epsilon$ and $\tau$) such that:
\begin{equation}
\label{Eq: Bound on Tail of Sum}
\sum_{m = K+1}^{\infty}{\exp\!\left(-2 L_m \gamma(\epsilon)^2 \right)} \leq \tau \, . 
\end{equation}
Note that such $K$ exists because $\sum_{m = 1}^{\infty}{1/m^2} = \pi^2 /6 < +\infty$, and for sufficiently large $m$, we have: 
\begin{equation}
\label{Eq:Relax Condition on L_m}
\exp\!\left(-2 L_m \gamma(\epsilon)^2\right) \leq \frac{1}{m^2} \quad \Leftrightarrow \quad -2 L_m \gamma(\epsilon)^2 \leq -2\log(m) \quad \Leftrightarrow \quad \frac{\log(m)}{L_m} \leq \gamma(\epsilon)^2 
\end{equation}
since $L_m = \omega(\log(m))$ (i.e. $\lim_{m \rightarrow \infty}{\log(m)/L_m} = 0$). Using the continuity of probability measures, observe that:
\begin{align*}
\P\!\left(\left. \bigcap_{k > K}{\left\{\sigma_k^+ \geq \hat{\sigma} - \epsilon\right\}} \, \right| \sigma_K^+ \geq \hat{\sigma} - \epsilon, \sigma_K^- \leq 1-\hat{\sigma} + \epsilon \right) & = \prod_{k > K}{\P\!\left(\left. \sigma_k^+ \geq \hat{\sigma} - \epsilon \, \right| \sigma_{k-1}^+ \geq \hat{\sigma} - \epsilon,A_{k,K}\right)} \\
& \geq \prod_{k > K}{1 - \exp\!\left(-2 L_k \gamma(\epsilon)^2\right)} \\
& \geq 1 - \sum_{k > K}{\exp\!\left(-2 L_k \gamma(\epsilon)^2\right)} \\
& \geq 1 - \tau
\end{align*}
where the first inequality follows from \eqref{Eq: Stability whp}, the second inequality is straightforward to establish using induction, and the final inequality follows from \eqref{Eq: Bound on Tail of Sum}. Therefore, we have for any $k > K$:
\begin{equation}
\label{Eq: Stability with Extra Conditioning}
\P\!\left(\left. \sigma_k^+ \geq \hat{\sigma} - \epsilon \, \right| \sigma_K^+ \geq \hat{\sigma} - \epsilon, \sigma_K^- \leq 1-\hat{\sigma} + \epsilon \right) \geq 1 - \tau \, . 
\end{equation}
Likewise, we can also prove mutatis mutandis that for any $k > K$:
\begin{equation}
\label{Eq: Stability with Extra Conditioning 2}
\P\!\left(\left. \sigma_k^- \leq 1-\hat{\sigma} + \epsilon \, \right| \sigma_K^+ \geq \hat{\sigma} - \epsilon, \sigma_K^- \leq 1-\hat{\sigma} + \epsilon \right) \geq 1 - \tau 
\end{equation}
where the choices of $\epsilon$, $\tau$, and $K$ in \eqref{Eq: Stability with Extra Conditioning 2} are the same as those in \eqref{Eq: Stability with Extra Conditioning} without loss of generality.

We need to show that $\limsup_{k \rightarrow \infty}{\P(\hat{S}_{k} \neq \sigma_0)} < \frac{1}{2}$, or equivalently, that there exists $\lambda > 0$ such that for all sufficiently large $k \in \N$:
\begin{align*}
\P\!\left(\hat{S}_k \neq \sigma_0\right) = \frac{1}{2}\P\!\left(\left.\hat{S}_k \neq \sigma_0 \, \right| \sigma_0 = 1\right) + \frac{1}{2}\P\!\left(\left.\hat{S}_k \neq \sigma_0 \, \right| \sigma_0 = 0\right) & \leq \frac{1 - \lambda}{2} \\
\Leftrightarrow \quad \P\!\left(\left.\sigma_k < \frac{1}{2} \, \right| \sigma_0 = 1\right) + \P\!\left(\left.\sigma_k \geq \frac{1}{2} \, \right| \sigma_0 = 0\right) & \leq 1 - \lambda \\
\Leftrightarrow \quad \P\!\left(\sigma_k^+ \geq \frac{1}{2}\right) - \P\!\left(\sigma_k^- \geq \frac{1}{2}\right) & \geq \lambda \, .
\end{align*}
To this end, let $E = \left\{\sigma_K^+ \geq \hat{\sigma} - \epsilon, \sigma_K^- \leq 1-\hat{\sigma} + \epsilon\right\}$, and observe that for all $k > K$:
\begin{align*}
\P\!\left(\sigma_k^+ \geq \frac{1}{2}\right) - \P\!\left(\sigma_k^- \geq \frac{1}{2}\right) & = \E\!\left[\I\!\left\{\sigma_k^+ \geq \frac{1}{2}\right\} - \I\!\left\{\sigma_k^- \geq \frac{1}{2}\right\} \right] \\
& \geq \E\!\left[\left(\I\!\left\{\sigma_k^+ \geq \frac{1}{2}\right\} - \I\!\left\{\sigma_k^- \geq \frac{1}{2}\right\}\right)\I\!\left\{E\right\}\right] \\
& = \E\!\left[\left. \I\!\left\{\sigma_k^+ \geq \frac{1}{2}\right\} - \I\!\left\{\sigma_k^- \geq \frac{1}{2}\right\} \right| E \right] \P(E) \\
& = \left(\P\!\left(\left. \sigma_k^+ \geq \frac{1}{2} \, \right| E \right) - \P\!\left(\left. \sigma_k^- \geq \frac{1}{2} \, \right| E \right)\right) \P(E) \\
& \geq \left(\P\!\left(\left. \sigma_k^+ \geq \hat{\sigma} - \epsilon \, \right| E \right) - \P\!\left(\left. \sigma_k^- > 1-\hat{\sigma} + \epsilon \, \right| E \right)\right) \P(E) \\
& \geq (1 - 2\tau) \P(E) \triangleq \lambda > 0
\end{align*}
where the first inequality holds because $\I\!\left\{\sigma_{k}^+ \geq \frac{1}{2}\right\} - \I\!\left\{\sigma_{k}^- \geq \frac{1}{2}\right\} \geq 0$ almost surely due to the monotonicity (property 3) of our Markovian coupling, the second inequality holds because $1 - \hat{\sigma} + \epsilon < \frac{1}{2} < \hat{\sigma} - \epsilon$ (since $\epsilon > 0$ is small), and the final inequality follows from \eqref{Eq: Stability with Extra Conditioning} and \eqref{Eq: Stability with Extra Conditioning 2}. This completes the proof for the $\delta \in \left(0,\frac{1}{6}\right)$ regime. 

We next prove that $\delta \in \left(\frac{1}{6},\frac{1}{2}\right)$ implies: 
\begin{equation}
\label{Eq: Strong Impossibility}
\lim_{k \rightarrow \infty}{\E\!\left[\left\|P_{\sigma_k|G}^+ - P_{\sigma_k|G}^-\right\|_{\sf{TV}}\right]} = 0 \, .
\end{equation}
Recall that we are considering the monotone Markovian coupling $\{(\sigma^+_k,\sigma^-_k) : k \in \N\}$ such that $\sigma^+_k \geq \sigma^-_k$ almost surely for every $k \in \N$. Then, it is also true that conditioned on an underlying random DAG $G$, $\sigma^+_k \geq \sigma^-_k$ almost surely for every $k \in \N$. Now observe that given a random DAG $G$:
$$ \left\|P_{\sigma_k|G}^+ - P_{\sigma_k|G}^-\right\|_{\sf{TV}} \leq \P\!\left(\sigma^+_k \neq \sigma^-_k \middle| G\right) = \P\!\left(\sigma^+_k - \sigma^-_k \geq \frac{1}{L_k} \middle| G\right) \leq L_k \, \E\!\left[\sigma^+_k - \sigma^-_k \middle| G\right] $$
where the first inequality follows from Dobrushin's maximal coupling representation of TV distance \cite[Chapter 4.2]{MarkovMixing}, the middle equality holds because the possible values of $\sigma^+_k$ and $\sigma^-_k$ are $\{m/L_k : m = 0,\dots,L_k\}$, and the final inequality follows from Markov's inequality since $\sigma^+_k - \sigma^-_k \geq 0$ almost surely given $G$. Taking expectations with respect to $G$, we have:
\begin{equation} 
\label{Eq: Initial TV Bound}
\E\!\left[\left\|P_{\sigma_k|G}^+ - P_{\sigma_k|G}^-\right\|_{\sf{TV}}\right] \leq L_k \, \E\!\left[\sigma^+_k - \sigma^-_k \right] .
\end{equation}
We can bound $\E\!\left[\sigma^+_k - \sigma^-_k\right]$ as follows. Firstly, we use the Lipschitz continuity of $g$ (with Lipschitz constant $\frac{3}{2}(1-2\delta)$) to get:
\begin{equation}
\label{Eq: Lipschitz expectation}
0 \leq \E\!\left[\left.\sigma^+_k - \sigma^-_k\right|\sigma^+_{k-1},\sigma^-_{k-1}\right] = g\!\left(\sigma^+_{k-1}\right) - g\!\left(\sigma^-_{k-1}\right) \leq \frac{3}{2}(1-2\delta) \left(\sigma^+_{k-1} - \sigma^-_{k-1}\right) \, .
\end{equation}
Then, we notice that:
\begin{align*}
0 \leq \E\!\left[\left.\sigma^+_k - \sigma^-_k\right|\sigma^+_{k-2},\sigma^-_{k-2}\right] & = \E\!\left[\left.\E\!\left[\left.\sigma^+_k - \sigma^-_k\right|\sigma^+_{k-1},\sigma^-_{k-1},\sigma^+_{k-2},\sigma^-_{k-2}\right]\right|\sigma^+_{k-2},\sigma^-_{k-2}\right] \\
& = \E\!\left[\left.\E\!\left[\left.\sigma^+_k - \sigma^-_k\right|\sigma^+_{k-1},\sigma^-_{k-1}\right]\right|\sigma^+_{k-2},\sigma^-_{k-2}\right] \\
& \leq \frac{3}{2}(1-2\delta) \, \E\!\left[\left.\sigma^+_{k-1} - \sigma^-_{k-1}\right|\sigma^+_{k-2},\sigma^-_{k-2}\right] \\
& \leq \left(\frac{3}{2}(1-2\delta)\right)^{\! 2} \! \left(\sigma^+_{k-2} - \sigma^-_{k-2}\right)
\end{align*}
where the first equality follows from the tower property, the second equality follows from the Markov property, and the final two inequalities both follow from \eqref{Eq: Lipschitz expectation}. Therefore, we recursively have:
$$ 0 \leq \E\!\left[\sigma^+_k - \sigma^-_k\right] \leq \left(\frac{3}{2}(1-2\delta)\right)^{\! k} $$
where we use the facts that $\E\!\left[\left.\sigma^+_k - \sigma^-_k\right|\sigma^+_{0},\sigma^-_{0}\right] = \E\!\left[\sigma^+_k - \sigma^-_k\right]$ and $\sigma^+_{0} - \sigma^-_{0} = 1$. Finally, using \eqref{Eq: Initial TV Bound}, we get:
$$ \E\!\left[\left\|P_{\sigma_k|G}^+ - P_{\sigma_k|G}^-\right\|_{\sf{TV}}\right] \leq L_k \, \left(\frac{3}{2}(1-2\delta)\right)^{\! k} $$
which in turn implies \eqref{Eq: Strong Impossibility} because $L_k = o((2/(3(1-2\delta)))^{k})$ by assumption. (It is worth mentioning that although $L_k = o((2/(3(1-2\delta)))^{k})$ in this regime, it can diverge to infinity because the Lipschitz constant $\frac{3}{2}(1-2\delta) < 1$.) This completes the proof.
\end{proof}

\renewcommand{\proofname}{\bfseries \emph{Proof}}

We remark that in the $\delta \in \left(0,\frac{1}{6}\right)$ regime, we can see from \eqref{Eq:Relax Condition on L_m} that $L_k \geq \log(k)/\gamma(\epsilon)^2$ for all sufficiently large $k$ suffices for the proof to hold, where a sufficiently small $\epsilon = \epsilon(\delta) > 0$ is fixed that depends on $\delta$ and ensures that $\gamma(\epsilon) > 0$. So, the condition that $L_k = \omega(\log(k))$ can be relaxed. 

In the next proposition, we show that if $L_k = \omega(\log(k))$, then the Markov chain $\{\sigma_k : k \in \N\}$ converges almost surely.

\begin{proposition}[Majority Random DAG Model Almost Sure Convergence]
\label{Prop: Majority Grid Almost Sure Convergence}
If $\delta \in \left(\delta_{\sf{maj}},\frac{1}{2}\right)$ and $L_k = \omega(\log(k))$, then $\lim_{k \rightarrow \infty}{\sigma_k} = \frac{1}{2}$ almost surely.
\end{proposition}

\begin{proof}
Recall that $L_k \sigma_k \sim \textsf{\small binomial}(L_k,g(\sigma))$ given $\sigma_{k-1} = \sigma$. This implies via Hoeffding's inequality and \eqref{Eq:Conditional Expectation} that for every $k \in \N\backslash\!\{0\}$ and $\epsilon_k > 0$:
$$ \P(|\sigma_k - g(\sigma_{k-1})| > \epsilon_k|\sigma_{k-1} = \sigma) \leq 2 \exp\!\left(-2 L_k \epsilon_k^2 \right) $$
where we can take expectations with respect to $\sigma_{k-1}$ to get:
\begin{equation}
\label{Eq:Hoeffding Consequence}
\P(|\sigma_k - g(\sigma_{k-1})| > \epsilon_k) \leq 2 \exp\!\left(-2 L_k \epsilon_k^2 \right) \, . 
\end{equation}
Now fix any $\tau > 0$, and choose a sufficiently large integer $K \in \N$ (that depends on $\tau$) such that:
$$ \P(\exists k > K, \, |\sigma_k - g(\sigma_{k-1})| > \epsilon_k) \leq \sum_{k = K+1}^{\infty}{\P(|\sigma_k - g(\sigma_{k-1})| > \epsilon_k)} \leq 2 \sum_{k = K+1}^{\infty}{\exp\!\left(-2 L_k \epsilon_k^2 \right)} \leq \tau $$
where we use the union bound, and let $\epsilon_k = \sqrt{\log(k)/L_k}$ %or use \geq
(or equivalently, $\exp\!\left(-2 L_k \epsilon_k^2 \right) = 1/k^2$). %or use \leq
This implies that for any $\tau > 0$:
\begin{equation}
\label{Eq:Precursor to Borel-Cantelli result}
\P(\forall k > K, \, |\sigma_k - g(\sigma_{k-1})| \leq \epsilon_k) \geq 1 - \tau \, . 
\end{equation}
Since for every $k > K$, $|\sigma_k - g(\sigma_{k-1})| \leq \epsilon_k$, we can recursively obtain the following relation:
\begin{equation}
\forall k \in \N\backslash[K+1], \enspace \left|\sigma_k - g^{(k-K)}(\sigma_K)\right| \leq \sum_{m = K+1}^{k}{c^{k-m}\epsilon_m} 
\end{equation} 
where $g^{(k-K)}$ denotes $g$ composed with itself $k-K$ times, and $c = \frac{3}{2}(1-2\delta) \in (0,1)$ (since $\delta \in \left(\frac{1}{6},\frac{1}{2}\right)$) denotes the Lipschitz constant of $g$ on $[0,1]$. Since $L_m = \omega(\log(m))$, for any $\epsilon > 0$, we can take $K = K(\epsilon,\tau)$ (which depends on both $\epsilon$ and $\tau$) to be sufficiently large so that $\sup_{m > K}{\epsilon_m} \leq \epsilon (1-c)$. Now observe that we have:
$$ \forall k \in \N\backslash[K+1], \enspace \sum_{m = K+1}^{k}{c^{k-m}\epsilon_m} \leq \left(\sup_{m > K}{\epsilon_m}\right)\! \sum_{j = 0}^{\infty}{c^{j}} = \left(\sup_{m > K}{\epsilon_m}\right)\! \frac{1}{1-c} \leq \epsilon \, . $$
Moreover, notice that $\lim_{m \rightarrow \infty}{g^{(m)}(\sigma_K)} = \frac{1}{2}$ by the fixed point theorem \cite[Chapter 5, Exercise 22(c)]{Rudin} (recalling that $\sigma = \frac{1}{2}$ is the only fixed point of $g$ when $\delta \in \left(\frac{1}{6},\frac{1}{2}\right)$). As a result, for any $\tau > 0$ and any $\epsilon > 0$, there exists $K = K(\epsilon,\tau) \in \N$ such that:
$$ \P\!\left(\forall k > K, \, \left|\sigma_k - g^{(k-K)}(\sigma_K)\right| \leq \epsilon\right) \geq 1 - \tau $$ 
which implies, after letting $k \rightarrow \infty$, that:
$$ \P\!\left(\frac{1}{2} - \epsilon \leq \liminf_{k \rightarrow \infty}{\sigma_k} \leq \limsup_{k \rightarrow \infty}{\sigma_k} \leq \frac{1}{2} + \epsilon\right) \geq 1 - \tau \, . $$
Lastly, we can first let $\epsilon \rightarrow 0$ and employ the continuity of $\P$, and then let $\tau \rightarrow 0$ to obtain:
$$ \P\!\left(\lim_{k \rightarrow \infty}{\sigma_k} = \frac{1}{2}\right) = 1 \, . $$
This completes the proof.
\end{proof}

Proposition \ref{Prop: Majority Grid Almost Sure Convergence} can be construed as a ``weak'' impossibility result since it demonstrates that the average number of $1$'s tends to $\frac{1}{2}$ in the $\delta \in \left(\frac{1}{6},\frac{1}{2}\right)$ regime regardless of initial state of the Markov chain $\{\sigma_k : k \in \N\}$.

\section{Analysis of AND-OR Rule Processing in Random DAG Model}
\label{Analysis of And-Or Rule Processing in Random Grid}

In this section, we prove Theorem \ref{Thm:Phase Transition in Random Grid with And-Or Rule Processing}. As before, we begin by making some pertinent observations. Recall that we have a random DAG model with $d = 2$, and all Boolean functions at even levels are the AND rule, and all Boolean functions at odd levels are the OR rule, i.e. $f_{k}(x_1,x_2) = x_1 \land x_2$ for every $k \in 2\N\backslash\!\{0\}$, and $f_{k}(x_1,x_2) = x_1 \lor x_2$ for every $k \in \N\backslash 2\N$. Suppose we are given that $\sigma_{k-1} = \sigma$. Then, for every $k \in \N\backslash\!\{0\}$ and every $j \in [L_{k}]$:
\begin{equation}
X_{k,j} = \left\{ 
\begin{array}{lcl}
\Ber(\sigma \star \delta) \land \Ber(\sigma \star \delta) & , & \text{if } k \text{ even} \\ 
\Ber(\sigma \star \delta) \lor \Ber(\sigma \star \delta) & , & \text{if } k \text{ odd}
\end{array} 
\right.
\end{equation}
for two i.i.d. Bernoulli random variables. Since we have:
\begin{align}
\P(X_{k,j} = 1|\sigma_{k-1} = \sigma) & = \left\{
\begin{array}{lcl}
(\sigma \star \delta)^2 & , & \text{if } k \text{ even} \\
1 - (1 - \sigma \star \delta)^2 & , & \text{if } k \text{ odd}
\end{array}
\right. \\
& = \E[\sigma_k|\sigma_{k-1} = \sigma] \, ,
\end{align}
$X_{k,j}$ are i.i.d. $\Ber(g_{k \, (\textsf{mod} \, 2)}(\sigma))$ for $j \in [L_{k}]$, and $L_k \sigma_k \sim \textsf{\small binomial}(L_k,g_{k \, (\textsf{mod} \, 2)}(\sigma))$, where we define $g_0:[0,1] \rightarrow [0,1]$ as $g_0(\sigma) \triangleq (\sigma \star \delta)^2$, and $g_1:[0,1] \rightarrow [0,1]$ as $g_1(\sigma) \triangleq 1 - (1 - \sigma \star \delta)^2 = 2(\sigma \star \delta) - (\sigma \star \delta)^2$. The derivatives of $g_0$ and $g_1$ are:
\begin{align}
\label{Eq: g_0 Derivative}
g_0^{\prime}(\sigma) & = 2(1-2\delta)(\sigma \star \delta) \geq 0 \, , \\
g_1^{\prime}(\sigma) & = 2(1 - 2\delta)(1-\sigma \star \delta) \geq 0 \, . 
\label{Eq: g_1 Derivative}
\end{align}
Consider the composition of $g_0$ and $g_1$, $g \triangleq g_0 \circ g_1:[0,1] \rightarrow [0,1]$, $g(\sigma) = \left(\left(2(\sigma \star \delta) - (\sigma \star \delta)^2\right) \star \delta\right)^2$, which has derivative $g^{\prime}:[0,1] \rightarrow \R^{+}$ given by:
\begin{align}
 g^{\prime}(\sigma) & = g_0^{\prime}(g_1(\sigma)) g_1^{\prime}(\sigma) \nonumber \\
& = 4(1-2\delta)^2 (g_1(\sigma) \star \delta) (1-\sigma \star \delta) \geq 0 \, .
\end{align}
This is a cubic function of $\sigma$ with maximum value:
\begin{align}
C(\delta) \triangleq \max_{\sigma \in [0,1]}{g^{\prime}(\sigma)} & = \left\{
\begin{array}{lcl}
g^{\prime}\!\left(\frac{1-\delta}{1-2\delta} - \sqrt{\frac{1-\delta}{3(1-2\delta)^3}}\right) & , & \delta \in \left(0,\frac{9-\sqrt{33}}{12}\right] \\
g^{\prime}(0) & , & \delta \in \left(\frac{9-\sqrt{33}}{12},\frac{1}{2}\right)
\end{array} \right. \\
& = \left\{
\begin{array}{lcl}
\left(\frac{4(1-\delta)(1-2\delta)}{3}\right)^{\frac{3}{2}} & , & \delta \in \left(0,\frac{9-\sqrt{33}}{12}\right] \\
4\delta(1-\delta)^2 (1-2\delta)^2 (3-2\delta) < 1 & , & \delta \in \left(\frac{9-\sqrt{33}}{12},\frac{1}{2}\right)
\end{array} \right.
\label{Eq: Lipschitz constant}
\end{align}
which follow from standard calculus and algebraic manipulations, and Wolfram Mathematica computations. Hence, $C(\delta)$ in \eqref{Eq: Lipschitz constant} is the Lipschitz constant of $g$ over $[0,1]$. Since $4(1-\delta)(1-2\delta)/3 \in (0,1) \Leftrightarrow \delta \in ((3-\sqrt{7})/4,(9 - \sqrt{33})/12]$, $C(\delta) < 1$ if and only if $\delta \in ((3-\sqrt{7})/4,1/2)$. Moreover, $C(\delta) > 1$ if and only if $\delta \in (0,(3-\sqrt{7})/4)$ (and $C(\delta) = 1$ when $\delta = (3-\sqrt{7})/4$).

We next summarize the fixed point structure of $g$. Solving the equation $g(\sigma) = \sigma$ in Wolfram Mathematica produces:
\begin{equation}
\sigma = \frac{1-6\delta + 4\delta^2 \pm \sqrt{1-12\delta + 8\delta^2}}{2(1-2\delta)^2},\frac{3-6\delta + 4\delta^2 \pm \sqrt{5-12\delta + 8\delta^2}}{2(1-2\delta)^2}
\end{equation}
where the first pair is real when $\delta \in [0,(3-\sqrt{7})/4]$, and the second pair is always real. From these solutions, it is straightforward to verify that the only fixed points of $g$ in the interval $[0,1]$ are:
\begin{align}
t_0 & \triangleq \frac{2(1-\delta)(1-2\delta) - 1 - \sqrt{4(1-\delta)(1-2\delta) - 3}}{2(1-2\delta)^2} \quad \text{(valid when $\delta \in \left(0,\frac{3-\sqrt{7}}{4}\right]$)} \\
t_1 & \triangleq \frac{2(1-\delta)(1-2\delta) - 1 + \sqrt{4(1-\delta)(1-2\delta) - 3}}{2(1-2\delta)^2} \quad \text{(valid when $\delta \in \left(0,\frac{3-\sqrt{7}}{4}\right]$)} \\
t & \triangleq \frac{2(1-\delta)(1-2\delta) + 1 - \sqrt{4(1-\delta)(1-2\delta) + 1}}{2(1-2\delta)^2} 
\label{Eq: Middle Fixed Point}
\end{align}
which satisfy $t_0 = t_1 = t$ when $\delta = (3-\sqrt{7})/4$, and $t_0 = 0, \, t_1 = 1$ when $\delta = 0$. Furthermore, observe that:
\begin{equation}
t_1 - t = \frac{\sqrt{a} + \sqrt{a + 4} - 2}{2(1-2\delta)^2} > 0  \quad \text{and} \quad t - t_0 = \frac{\sqrt{a} - \sqrt{a+4} + 2}{2(1-2\delta)^2} > 0 
\end{equation}
where $a = 4(1-\delta)(1-2\delta) - 3 > 0$ for $\delta \in (0,(3-\sqrt{7})/4)$, $t_1 - t > 0$ because $x \mapsto \sqrt{x}$ is strictly increasing ($\Rightarrow \sqrt{a} + \sqrt{a + 4} > 2$), and $t - t_0 > 0$ because $x \mapsto \sqrt{x}$ is strictly subadditive ($\Rightarrow \sqrt{a} + 2 > \sqrt{a + 4}$). Hence, $0 < t_0 < t < t_1 < 1$ when $\delta \in (0,(3-\sqrt{7})/4)$.

Therefore, there are again two regimes of interest. In the regime $\delta \in (0,(3-\sqrt{7})/4)$, $g$ has three fixed points $0 < t_0 < t < t_1 < 1$, and $C(\delta) > 1$. In contrast, in the regime $\delta \in ((3-\sqrt{7})/4,1/2)$, $g$ has only one fixed point at $t \in (0,1)$, and $C(\delta) < 1$. We now prove Theorem \ref{Thm:Phase Transition in Random Grid with And-Or Rule Processing}. 

\renewcommand{\proofname}{\bfseries \emph{Proof of Theorem \ref{Thm:Phase Transition in Random Grid with And-Or Rule Processing}}}

\begin{proof}
The proof of this result closely resembles the proof of Theorem \ref{Thm:Phase Transition in Random Grid with Majority Rule Processing} in section \ref{Analysis of Majority Rule Processing in Random Grid}. 

We first prove that $\delta \in \left(0,(3-\sqrt{7})/4\right)$ implies $\limsup_{k \rightarrow \infty}{\P(\hat{S}_{2k} \neq \sigma_0)} < \frac{1}{2}$. As before, since $g_0$ and $g_1$ are non-decreasing functions, we can construct a monotone Markovian coupling $\{(\sigma^+_{k},\sigma^-_{k}) : k \in \N\}$ between $\{\sigma^+_{k} : k \in \N\}$ and $\{\sigma^-_{k} : k \in \N\}$, which are versions of the Markov chain $\{\sigma_{k} : k \in \N\}$ initialized at $\sigma^+_0 = 1$ and $\sigma^-_0 = 0$, respectively. This Markovian coupling satisfies the following properties:
\begin{itemize}
\item[1)] The ``marginal'' Markov chains are $\{\sigma^+_{k} : k \in \N\}$ and $\{\sigma^-_{k} : k \in \N\}$.
\item[2)] For every $j > k \geq 1$, $\sigma_{j}^+$ is conditionally independent of $\sigma^-_0,\dots,\sigma^-_{k},\sigma^+_0,\dots,\sigma^+_{k-1}$ given $\sigma^+_{k}$, and $\sigma_{j}^-$ is conditionally independent of $\sigma^+_0,\dots,\sigma^+_{k},\sigma^-_0,\dots,\sigma^-_{k-1}$ given $\sigma^-_{k}$. 
\item[3)] For every $k \in \N$, $\sigma^+_{k} \geq \sigma^-_{k}$ almost surely.
\end{itemize}
We next prove that there exists $\epsilon > 0$ such that:
\begin{equation}
\label{Eq: Stability whp 2}
\forall k \in \N\backslash\!\{0\}, \enspace \P\!\left(\left.\sigma_{2k}^+ \geq t_1 - \epsilon \, \right| \sigma_{2k-2}^+ \geq t_1 - \epsilon,A_{k,j}\right) \geq 1 - 4 \exp\!\left(-\frac{(L_{2k} \wedge L_{2k-1}) \gamma(\epsilon)^2}{8}\right)
\end{equation}
where $\gamma(\epsilon) \triangleq g(t_1 - \epsilon) - (t_1 - \epsilon) > 0$, and $A_{k,j}$ with $0 \leq j < k$ is the non-zero probability event defined as: 
$$ A_{k,j} \triangleq \left\{ 
\begin{array}{lcl}
\{\sigma_{2j}^- \leq t_0 + \epsilon\} & , & 0 \leq j = k-1 \\
\{\sigma_{2k-4}^+ \geq t_1 - \epsilon,\sigma_{2k-6}^+ \geq t_1 - \epsilon,\dots,\sigma_{2j}^+ \geq t_1 - \epsilon\} \cap \{\sigma_{2j}^- \leq t_0 + \epsilon\} & , & 0 \leq j \leq k-2
\end{array} \right. \, .
$$ 
Since $g^{\prime}(t_1) = 4\delta(3-2\delta) < 1$ and $g(t_1) = t_1$, $g(t_1 - \epsilon) > t_1 - \epsilon$ for sufficiently small $\epsilon > 0$. Fix any such $\epsilon > 0$ such that $\gamma(\epsilon) > 0$. Observe that for every $k \in \N\backslash\!\{0\}$ and $\xi > 0$, we have:
\begin{align}
& \P(|\sigma_{2k} - g(\sigma_{2k-2})| > \xi | \sigma_{2k-2} = \sigma) \nonumber \\
& \leq \P(|\sigma_{2k} - g_0(\sigma_{2k-1})| + |g_0(\sigma_{2k-1}) - g_0(g_1(\sigma_{2k-2}))| > \xi | \sigma_{2k-2} = \sigma) \nonumber \\
& \leq \P(|\sigma_{2k} - g_0(\sigma_{2k-1})| + 2(1-\delta)(1-2\delta)|\sigma_{2k-1} - g_1(\sigma_{2k-2})| > \xi | \sigma_{2k-2} = \sigma) \nonumber \\
& \leq \P\!\left(\left. \left\{|\sigma_{2k} - g_0(\sigma_{2k-1})| > \frac{\xi}{2}\right\} \cup \left\{2(1-\delta)(1-2\delta)|\sigma_{2k-1} - g_1(\sigma_{2k-2})| > \frac{\xi}{2}\right\} \right| \sigma_{2k-2} = \sigma\right) \nonumber \\
& \leq \P\!\left(\left.|\sigma_{2k} - g_0(\sigma_{2k-1})| > \frac{\xi}{2}\right| \sigma_{2k-2} = \sigma\right) + \P\!\left(\left.|\sigma_{2k-1} - g_1(\sigma_{2k-2})| > \frac{\xi}{4(1-\delta)(1-2\delta)}\right| \sigma_{2k-2} = \sigma\right) \nonumber \\
& \leq \E\!\left[\left.\P\!\left(\left.|\sigma_{2k} - g_0(\sigma_{2k-1})| > \frac{\xi}{2}\right|\sigma_{2k-1}\right) \right| \sigma_{2k-2} = \sigma\right] + 2 \exp\!\left(-\frac{L_{2k-1} \xi^2}{8(1-\delta)^2 (1-2\delta)^2} \right) \nonumber \\
& \leq 2 \exp\!\left(-\frac{L_{2k} \xi^2}{2}\right) + 2 \exp\!\left(-\frac{L_{2k-1} \xi^2}{8(1-\delta)^2 (1-2\delta)^2} \right) \nonumber \\
& \leq 4 \exp\!\left(-\frac{(L_{2k} \wedge L_{2k-1}) \xi^2}{8}\right)
\label{Eq: Hoeffding Consequence 2}
\end{align} 
where the first inequality follows from the triangle inequality and the fact that $g = g_0 \circ g_1$, the second inequality holds because the Lipschitz constant of $g_0$ on $[0,1]$ is $\max_{\sigma \in [0,1]}{g_0^{\prime}(\sigma)} = g_0^{\prime}(1) = 2(1-\delta)(1-2\delta)$ using \eqref{Eq: g_0 Derivative}, the fourth inequality follows from the union bound, the fifth and sixth inequalities follow from the Markov property and Hoeffding's inequality (as well as the fact that $L_k \sigma_k \sim \textsf{\small binomial}(L_k,g_{k \, (\textsf{mod} \, 2)}(\sigma))$ given $\sigma_{k-1} = \sigma$), the final inequality holds because $(1-\delta)^2 (1-2\delta)^2 \leq 1$, and $\wedge$ denotes the minimum operation. Hence, for any $k \in \N\backslash\!\{0\}$ and any $0 \leq j < k$, we have:
\begin{align*}
\P\!\left(\left.\sigma_{2k}^+ < g\!\left(\sigma_{2k-2}^+\right) - \gamma(\epsilon)\, \right|\sigma_{2k-2}^+ = \sigma,A_{k,j}\right) & = \P\!\left(\left.\sigma_{2k} < g\!\left(\sigma_{2k-2}\right) - \gamma(\epsilon)\, \right|\sigma_{2k-2} = \sigma\right) \\
& \leq \P(|\sigma_{2k} - g(\sigma_{2k-2})| > \gamma(\epsilon) | \sigma_{2k-2} = \sigma) \\
& \leq 4 \exp\!\left(-\frac{(L_{2k} \wedge L_{2k-1}) \gamma(\epsilon)^2}{8}\right)
\end{align*} 
where the first equality follows from property 2 of the Markovian coupling, and the final inequality follows from \eqref{Eq: Hoeffding Consequence 2}. As shown in the proof of Theorem \ref{Thm:Phase Transition in Random Grid with Majority Rule Processing}, this produces:
\begin{align*}
\P\!\left(\left.\sigma_{2k}^+ < g\!\left(\sigma_{2k-2}^+\right) - \gamma(\epsilon) \, \right| \sigma_{2k-2}^+ \geq t_1 - \epsilon, A_{k,j}\right) & \leq 4 \exp\!\left(-\frac{(L_{2k} \wedge L_{2k-1}) \gamma(\epsilon)^2}{8}\right) \\
\P\!\left(\left.\sigma_{2k}^+ < t_1 - \epsilon \, \right| \sigma_{2k-2}^+ \geq t_1 - \epsilon, A_{k,j}\right) & \leq 4 \exp\!\left(-\frac{(L_{2k} \wedge L_{2k-1}) \gamma(\epsilon)^2}{8}\right)
\end{align*}
where the second inequality follows from the first because $\sigma_{2k}^+ < t_1 - \epsilon = g(t_1 - \epsilon) - \gamma(\epsilon)$ implies that $\sigma_{2k}^+ < g(\sigma_{2k-2}^+) - \gamma(\epsilon)$ when $\sigma_{2k-2}^+ \geq t_1 - \epsilon$ (since $g$ is non-decreasing and $g(\sigma_{2k-2}^+) \geq g(t_1 - \epsilon)$). This proves \eqref{Eq: Stability whp 2}.

Now fix any $\tau > 0$, and choose a sufficiently large even integer $K = K(\epsilon,\tau) \in 2\N$ (that depends on $\epsilon$ and $\tau$) such that:
\begin{equation}
\label{Eq: Bound on Tail of Sum 2}
4 \sum_{m = \frac{K}{2}+1}^{\infty}{\exp\!\left(-\frac{(L_{2m} \wedge L_{2m-1}) \gamma(\epsilon)^2}{8}\right)} \leq \tau \, . 
\end{equation}
Note that such $K$ exists because $\sum_{m = 1}^{\infty}{1/m^2} = \pi^2 /6 < +\infty$, and for sufficiently large $m$, we have: 
$$ \exp\!\left(-\frac{(L_{2m} \wedge L_{2m-1}) \gamma(\epsilon)^2}{8}\right) \leq \frac{1}{m^2} $$
since $L_m = \omega(\log(m))$. Using the continuity of probability measures, observe that:
\begin{align*}
\P\!\left(\left. \bigcap_{k > \frac{K}{2}}{\left\{\sigma_{2k}^+ \geq t_1 - \epsilon\right\}} \, \right| \sigma_{K}^+ \geq t_1 - \epsilon, \sigma_{K}^- \leq t_0 + \epsilon \right) & = \prod_{k > \frac{K}{2}}{\P\!\left(\sigma_{2k}^+ \geq t_1 - \epsilon  \left| \, \sigma_{2k-2}^+ \geq t_1 - \epsilon,A_{k,\frac{K}{2}} \right.\right)} \\
& \geq \prod_{k > \frac{K}{2}}{1 - 4 \exp\!\left(- \frac{(L_{2k} \wedge L_{2k-1}) \gamma(\epsilon)^2}{8}\right)} \\
& \geq 1 - 4 \sum_{k > \frac{K}{2}}{\exp\!\left(- \frac{(L_{2k} \wedge L_{2k-1}) \gamma(\epsilon)^2}{8}\right)} \\
& \geq 1 - \tau
\end{align*}
where the first inequality follows from \eqref{Eq: Stability whp 2}, and the final inequality follows from \eqref{Eq: Bound on Tail of Sum 2}. Therefore, we have for any $k > \frac{K}{2}$:
\begin{equation}
\label{Eq: And-Or Stability with Extra Conditioning}
\P\!\left(\left. \sigma_{2k}^+ \geq t_1 - \epsilon \, \right| \sigma_K^+ \geq t_1 - \epsilon, \sigma_K^- \leq t_0 + \epsilon \right) \geq 1 - \tau \, . 
\end{equation}
Likewise, we can also prove mutatis mutandis that for any $k > \frac{K}{2}$:
\begin{equation}
\label{Eq: And-Or Stability with Extra Conditioning 2}
\P\!\left(\left. \sigma_{2k}^- \leq t_0 + \epsilon \, \right| \sigma_K^+ \geq t_1 - \epsilon, \sigma_K^- \leq t_0 + \epsilon \right) \geq 1 - \tau 
\end{equation}
where $\epsilon$, $\tau$, and $K$ in \eqref{Eq: And-Or Stability with Extra Conditioning 2} can be chosen to be the same as those in \eqref{Eq: And-Or Stability with Extra Conditioning} without loss of generality.

Finally, we let $E = \left\{\sigma_K^+ \geq t_1 - \epsilon, \sigma_K^- \leq t_0 + \epsilon\right\}$, and observe that for all $k > \frac{K}{2}$:
\begin{align*}
\P\!\left(\sigma_{2k}^+ \geq t\right) - \P\!\left(\sigma_{2k}^- \geq t\right) & \geq \E\!\left[\left(\I\!\left\{\sigma_{2k}^+ \geq t\right\} - \I\!\left\{\sigma_{2k}^- \geq t\right\}\right)\I\!\left\{E\right\}\right] \\
& = \left(\P\!\left(\left. \sigma_{2k}^+ \geq t \, \right| E \right) - \P\!\left(\left. \sigma_{2k}^- \geq t \, \right| E \right)\right) \P(E) \\
& \geq \left(\P\!\left(\left. \sigma_{2k}^+ \geq t_1 - \epsilon \, \right| E \right) - \P\!\left(\left. \sigma_{2k}^- > t_0 + \epsilon \, \right| E \right)\right) \P(E) \\
& \geq (1 - 2\tau) \P(E) > 0
\end{align*}
where the first inequality holds because $\I\!\left\{\sigma_{2k}^+ \geq t\right\} - \I\!\left\{\sigma_{2k}^- \geq t\right\} \geq 0$ almost surely due to the monotonicity (property 3) of our Markovian coupling, the second inequality holds because $t_0 + \epsilon < t < t_1 - \epsilon$ (since $\epsilon > 0$ is small), and the final inequality follows from \eqref{Eq: And-Or Stability with Extra Conditioning} and \eqref{Eq: And-Or Stability with Extra Conditioning 2}. As argued in the proof of Theorem \ref{Thm:Phase Transition in Random Grid with Majority Rule Processing}, this illustrates that $\limsup_{k \rightarrow \infty}{\P(\hat{T}_{2k} \neq \sigma_0)} < \frac{1}{2}$. 

We next prove that $\delta \in ((3-\sqrt{7})/4,1/2)$ implies:
\begin{equation}
\label{Eq: Strong Impossibility 2}
\lim_{k \rightarrow \infty}{\E\!\left[\left\|P_{\sigma_{2k}|G}^+ - P_{\sigma_{2k}|G}^-\right\|_{\sf{TV}}\right]} = 0 \, . 
\end{equation}
In the regime where $L_k = o(\log(k))$, we have the desired result due to part 2 of Proposition \ref{Prop: Slow Growth of Layers} in subsection \ref{Further Discussion}. So, we will assume that $L_k \rightarrow \infty$ as $k \rightarrow \infty$ in the remaining proof. Observe that given a random DAG $G$:
$$ \left\|P_{\sigma_{2k}|G}^+ - P_{\sigma_{2k}|G}^-\right\|_{\sf{TV}} \leq \P\!\left(\sigma^+_{2k} \neq \sigma^-_{2k}\middle| G\right) = \P\!\left(\sigma^+_{2k} - \sigma^-_{2k} \geq \frac{1}{L_{2k}} \middle| G\right) \leq L_{2k} \, \E\!\left[\sigma^+_{2k} - \sigma^-_{2k}\middle| G\right] $$
where we use the monotone Markovian coupling defined earlier which ensures that $\sigma_{2k}^+ \geq \sigma_{2k}^-$ almost surely for every $k \in \N$ given $G$, the first inequality uses the maximal coupling representation of TV distance, and the final inequality follows from Markov's inequality. Taking expectations with respect to $G$, we have:
\begin{equation} 
\label{Eq: Initial TV Bound 2}
\E\!\left[\left\|P_{\sigma_{2k}|G}^+ - P_{\sigma_{2k}|G}^-\right\|_{\sf{TV}}\right] \leq L_{2k} \, \E\!\left[\sigma^+_{2k} - \sigma^-_{2k}\right] \, .
\end{equation}
We can bound $\E\!\left[\sigma^+_{2k} - \sigma^-_{2k}\right]$ as follows. Firstly, for any $k \in \N\backslash\!\{0\}$, we have:
\begin{align}
\E\!\left[\left.\sigma^+_{2k} - \sigma^-_{2k}\right|\sigma^+_{2k-2},\sigma^-_{2k-2}\right] & = \E\!\left[\left.\E\!\left[\left.\sigma^+_{2k} - \sigma^-_{2k}\right|\sigma^+_{2k-1},\sigma^-_{2k-1}\right]\right|\sigma^+_{2k-2},\sigma^-_{2k-2}\right] \nonumber \\
& = \E\!\left[\left.g_0\!\left(\sigma^+_{2k-1}\right) - g_0\!\left(\sigma^-_{2k-1}\right)\right|\sigma^+_{2k-2},\sigma^-_{2k-2}\right] 
\label{Eq: g0 Difference}
\end{align}
where the first equality follows from the tower and Markov properties, and the second equality holds because $L_{2k} \sigma_{2k} \sim \textsf{\small binomial}(L_{2k},g_{0}(\sigma))$ given $\sigma_{2k-1} = \sigma$. Then, recalling that $g_0(\sigma) = (\sigma \star \delta)^2 = (1-2\delta)^2 \sigma^2 + 2 \delta (1-2\delta) \sigma + \delta^2$, we can compute:
\begin{align}
\E\!\left[\left.g_0\!\left(\sigma^+_{2k-1}\right)\right|\sigma^+_{2k-2},\sigma^-_{2k-2}\right] & = \E\!\left[\left.g_0\!\left(\sigma^+_{2k-1}\right)\right|\sigma^+_{2k-2}\right] \nonumber \\
& = \E\!\left[\left.(1-2\delta)^2 \sigma^{+\quad 2}_{2k-1} + 2 \delta (1-2\delta) \sigma^+_{2k-1} + \delta^2 \right|\sigma^+_{2k-2}\right] \nonumber \\
& = (1-2\delta)^2 \! \left(\VAR\!\left(\left.\sigma^{+}_{2k-1} \right|\sigma^+_{2k-2}\right) + \E\!\left[\left.\sigma^{+}_{2k-1} \right|\sigma^+_{2k-2}\right]^2\right) \nonumber \\
& \quad \, + 2 \delta (1-2\delta) \E\!\left[\left.\sigma^{+}_{2k-1} \right|\sigma^+_{2k-2}\right] + \delta^2 \nonumber \\
& = (1-2\delta)^2 g_1\!\left(\sigma^+_{2k-2}\right)^2 + 2 \delta (1-2\delta) g_1\!\left(\sigma^+_{2k-2}\right) + \delta^2 \nonumber \\
& \quad \, + (1-2\delta)^2 \, \frac{g_1\!\left(\sigma^+_{2k-2}\right) \! \left(1 - g_1\!\left(\sigma^+_{2k-2}\right)\right)}{L_{2k-1}} \nonumber \\
& = g\!\left(\sigma^+_{2k-2}\right) + (1-2\delta)^2 \, \frac{g_1\!\left(\sigma^+_{2k-2}\right) \! \left(1 - g_1\!\left(\sigma^+_{2k-2}\right)\right)}{L_{2k-1}} 
\label{Eq: g0 Expectation}
\end{align}
where the first equality uses property 2 of the monotone Markovian coupling, and the fourth equality uses the fact that $L_{2k-1} \sigma_{2k-1} \sim \textsf{\small binomial}(L_{2k-1},g_{1}(\sigma))$ given $\sigma_{2k-2} = \sigma$. Using \eqref{Eq: g0 Difference} and \eqref{Eq: g0 Expectation}, we get:
\begin{align}
\E\!\left[\left.\sigma^+_{2k} - \sigma^-_{2k}\right|\sigma^+_{2k-2},\sigma^-_{2k-2}\right] & = g\!\left(\sigma^+_{2k-2}\right) - g\!\left(\sigma^-_{2k-2}\right) \nonumber \\
& \quad + (1-2\delta)^2 \! \left(\frac{g_1\!\left(\sigma^+_{2k-2}\right) \! \left(1 - g_1\!\left(\sigma^+_{2k-2}\right)\right) - g_1\!\left(\sigma^-_{2k-2}\right) \! \left(1 - g_1\!\left(\sigma^-_{2k-2}\right)\right)}{L_{2k-1}}\right) \nonumber \\
& = g\!\left(\sigma^+_{2k-2}\right) - g\!\left(\sigma^-_{2k-2}\right) \nonumber \\
& \quad + (1-2\delta)^2 \! \left(\frac{g_1\!\left(\sigma^+_{2k-2}\right) - g_1\!\left(\sigma^-_{2k-2}\right) - \left(g_1\!\left(\sigma^+_{2k-2}\right)^2 - g_1\!\left(\sigma^-_{2k-2}\right)^2\right)}{L_{2k-1}}\right) \nonumber \\
& \leq g\!\left(\sigma^+_{2k-2}\right) - g\!\left(\sigma^-_{2k-2}\right) + (1-2\delta)^2 \! \left(\frac{g_1\!\left(\sigma^+_{2k-2}\right) - g_1\!\left(\sigma^-_{2k-2}\right)}{L_{2k-1}}\right) \nonumber \\
& \leq \left(C(\delta) + \frac{2(1 - \delta)(1-2\delta)^3}{L_{2k-1}}\right) \! \left(\sigma^+_{2k-2} - \sigma^-_{2k-2}\right) \nonumber \\
& \leq \left(C(\delta) + \frac{2}{L_{2k-1}}\right) \! \left(\sigma^+_{2k-2} - \sigma^-_{2k-2}\right)
\label{Eq: Lipschitz expectation 2}
\end{align}
where the first inequality holds because $g_1\!\left(\sigma^+_{2k-2}\right)^2 - g_1\!\left(\sigma^-_{2k-2}\right)^2 \geq 0$ almost surely (since $g_1$ is non-negative and non-decreasing by \eqref{Eq: g_1 Derivative}, and $\sigma^+_{2k-2} \geq \sigma^-_{2k-2}$ almost surely by property 3 of the monotone Markovian coupling), the second inequality holds because $\sigma^+_{2k-2} \geq \sigma^-_{2k-2}$ almost surely and $g$ and $g_1$ have Lipschitz constants $C(\delta)$ and $\max_{\sigma \in [0,1]}{g_1^{\prime}(\sigma)} = 2(1 - \delta)(1 - 2\delta)$ respectively, and the final inequality holds because $(1 - \delta)(1-2\delta)^3 \leq 1$. 

Then, as in the proof of Theorem \ref{Thm:Phase Transition in Random Grid with Majority Rule Processing}, we notice using \eqref{Eq: Lipschitz expectation 2} that:
\begin{align*}
0 \leq \E\!\left[\left.\sigma^+_{2k} - \sigma^-_{2k}\right|\sigma^+_{2k-4},\sigma^-_{2k-4}\right] & = \E\!\left[\left.\E\!\left[\left.\sigma^+_{2k} - \sigma^-_{2k}\right|\sigma^+_{2k-2},\sigma^-_{2k-2}\right]\right|\sigma^+_{2k-4},\sigma^-_{2k-4}\right] \\
& \leq \left(C(\delta) + \frac{2}{L_{2k-1}}\right) \E\!\left[\left.\sigma^+_{2k-2} - \sigma^-_{2k-2}\right|\sigma^+_{2k-4},\sigma^-_{2k-4}\right] \\
& \leq \left(C(\delta) + \frac{2}{L_{2k-1}}\right) \! \left(C(\delta) + \frac{2}{L_{2k-3}}\right) \! \left(\sigma^+_{2k-4} - \sigma^-_{2k-4}\right)
\end{align*}
which recursively produces:
$$ 0 \leq \E\!\left[\sigma^+_{2k} - \sigma^-_{2k}\right] \leq \prod_{i = 1}^{k}{\left(C(\delta) + \frac{2}{L_{2i-1}}\right)} $$
where we use the facts that $\E\!\left[\left.\sigma^+_{2k} - \sigma^-_{2k}\right|\sigma^+_{0},\sigma^-_{0}\right] = \E\!\left[\sigma^+_{2k} - \sigma^-_{2k}\right]$ and $\sigma^+_{0} - \sigma^-_{0} = 1$. Finally, using \eqref{Eq: Initial TV Bound 2}, we get:
\begin{equation}
\label{Eq: Final TV Bound}
\E\!\left[\left\|P_{\sigma_{2k}|G}^+ - P_{\sigma_{2k}|G}^-\right\|_{\sf{TV}}\right] \leq L_{2k} \, \prod_{i = 1}^{k}{\left(C(\delta) + \frac{2}{L_{2i-1}}\right)} \, . 
\end{equation}

Recall that $L_k = o\big((C(\delta) + \epsilon)^{-\frac{k}{2}}\big)$ for some $\epsilon \in (0,1-C(\delta))$ (that can depend on $\delta$), and furthermore, we can assume that $L_k \rightarrow \infty$ as $k \rightarrow \infty$. Hence, there exists $K = K(\epsilon) \in \N$ (that depends on $\epsilon$) such that for all $i > K$, $L_{2i-1} > \frac{2}{\epsilon}$. This means that we can further upper bound \eqref{Eq: Final TV Bound} as follows:
$$ \forall k > K, \enspace \E\!\left[\left\|P_{\sigma_{2k}|G}^+ - P_{\sigma_{2k}|G}^-\right\|_{\sf{TV}}\right] \leq L_{2k} (C(\delta) + \epsilon)^{k - K} \prod_{i = 1}^{K}{\left(C(\delta) + \frac{2}{L_{2i-1}}\right)} \, . $$
Letting $k \rightarrow \infty$ produces \eqref{Eq: Strong Impossibility 2}. This completes the proof.
\end{proof}

\renewcommand{\proofname}{\bfseries \emph{Proof}}

The remarks after the proof of Theorem \ref{Thm:Phase Transition in Random Grid with Majority Rule Processing} in section \ref{Analysis of Majority Rule Processing in Random Grid} also apply here. In particular, in the $\delta \in (0,(3-\sqrt{7})/4)$ regime, the $L_k = \omega(\log(k))$ condition can be relaxed to $L_k \geq h(\delta) \log(k)$ for all sufficiently large $k$, where $h(\delta)$ is some fixed constant that depends on $\delta$. 

In the next proposition, we show that if $L_k = \omega(\log(k))$, then the Markov chain $\{\sigma_{2k} : k \in \N\}$ converges almost surely.

\begin{proposition}[AND-OR Random DAG Model Almost Sure Convergence]
\label{Prop: And-Or Grid Almost Sure Convergence}
If $\delta \in \left(\delta_{\sf{andor}},\frac{1}{2}\right)$ and $L_k = \omega(\log(k))$, then $\lim_{k \rightarrow \infty}{\sigma_{2k}} = t$ almost surely.
\end{proposition}

\begin{proof} 
This proof is analogous to the proof of Proposition \ref{Prop: Majority Grid Almost Sure Convergence}. For every $k \in \N\backslash\!\{0\}$ and $\epsilon_k > 0$, we have after taking expectations in \eqref{Eq: Hoeffding Consequence 2} that:
\begin{equation}
\label{Eq: Hoeffding Consequence 3}
\P(|\sigma_{2k} - g(\sigma_{2k-2})| > \epsilon_k) \leq 4 \exp\!\left(-\frac{(L_{2k} \wedge L_{2k-1}) \epsilon_k^2}{8}\right) \, . 
\end{equation}
Now fix any $\tau > 0$, and choose a sufficiently large integer $K \in \N$ (that depends on $\tau$) such that:
\begin{align*}
\P(\exists k \in \N, k > K, \, |\sigma_{2k} - g(\sigma_{2k-2})| > \epsilon_k) & \leq \sum_{m = K+1}^{\infty}{\P(|\sigma_{2m} - g(\sigma_{2m-2})| > \epsilon_{m})} \\
& \leq 4 \sum_{m = K+1}^{\infty}{\exp\!\left(-\frac{(L_{2m} \wedge L_{2m-1}) \epsilon_m^2}{8}\right)} \leq \tau 
\end{align*}
where we use the union bound and \eqref{Eq: Hoeffding Consequence 3}, and we set $\epsilon_m = 4\sqrt{\log(m)/(L_{2m} \wedge L_{2m-1})}$ (which ensures that $\exp(-(L_{2m} \wedge L_{2m-1}) \epsilon_m^2/8) = 1/m^2$). This implies that for any $\tau > 0$:
\begin{equation}
\P(\forall k \in \N, k > K, \, |\sigma_{2k} - g(\sigma_{2k-2})| \leq \epsilon_k) \geq 1 - \tau \, . 
\end{equation}
Since for every $k > K$, $|\sigma_{2k} - g(\sigma_{2k-2})| \leq \epsilon_k$, we can recursively obtain the following relation:
\begin{equation}
\forall k \in \N, k > K, \enspace \left|\sigma_{2k} - g^{(k-K)}(\sigma_{2K})\right| \leq \sum_{m = K+1}^{k}{C(\delta)^{k-m}\epsilon_{m}} 
\end{equation} 
where $g^{(k-K)}$ denotes $g$ composed with itself $k-K$ times, and $C(\delta)$ denotes the Lipschitz constant of $g$ on $[0,1]$ as shown in \eqref{Eq: Lipschitz constant}, which is in $(0,1)$ since $\delta \in ((3-\sqrt{7})/4,1/2)$. Since $L_m = \omega(\log(m))$, for any $\epsilon > 0$, we can take $K = K(\epsilon,\tau) \in \N$ (which depends on both $\epsilon$ and $\tau$) to be sufficiently large so that $\sup_{m > K}{\epsilon_{m}} \leq \epsilon (1-C(\delta))$. This implies that:
$$ \forall k \in \N, k > K, \enspace \sum_{m = K+1}^{k}{C(\delta)^{k-m}\epsilon_{m}} \leq \epsilon $$
as shown in the proof of Proposition \ref{Prop: Majority Grid Almost Sure Convergence}. Moreover, since $g:[0,1] \rightarrow [0,1]$ is a contraction when $\delta \in ((3-\sqrt{7})/4,1/2)$, it has a unique fixed point $t \in [0,1]$ by the fixed point theorem, and $\lim_{m \rightarrow \infty}{g^{(m)}(\sigma_K)} = t$ almost surely. As a result, for any $\tau > 0$ and any $\epsilon > 0$, there exists $K = K(\epsilon,\tau) \in \N$ such that:
$$ \P\!\left(\forall k \in \N, k > K, \, \left|\sigma_{2k} - g^{(k-K)}(\sigma_{2K})\right| \leq \epsilon\right) \geq 1 - \tau $$ 
which implies, after letting $k \rightarrow \infty$, that:
$$ \P\!\left(t - \epsilon \leq \liminf_{k \rightarrow \infty}{\sigma_{2k}} \leq \limsup_{k \rightarrow \infty}{\sigma_{2k}} \leq t + \epsilon\right) \geq 1 - \tau \, . $$
Lastly, we can first let $\epsilon \rightarrow 0$ and employ the continuity of $\P$, and then let $\tau \rightarrow 0$ to obtain:
$$ \P\!\left(\lim_{k \rightarrow \infty}{\sigma_{2k}} = t\right) = 1 \, . $$
This completes the proof.
\end{proof}

Much like Proposition \ref{Prop: Majority Grid Almost Sure Convergence}, Proposition \ref{Prop: And-Or Grid Almost Sure Convergence} can also be construed as a ``weak'' impossibility result.

\section{Analysis of Deterministic AND 2D Grid}
\label{Analysis of Deterministic And Grid}

We now turn to proving Theorem \ref{Thm: Deterministic And Grid}. Recall that we are given a deterministic 2D grid where all Boolean processing functions with two inputs are the AND rule, and all Boolean processing functions with one input are the identity rule, i.e. $f_{1}(x_1,x_2) = x_1 \wedge x_2$ and $f_{2}(x) = x$. 

As in our proofs with random DAG models, we construct a ``monotone'' Markovian coupling of the Markov chains $\{X^+_k : k \in \N\}$ and $\{X^-_k : k \in \N\}$, which denote versions of the Markov chain $\{X_k : k \in \N\}$ initialized at $X^+_0 = 1$ and $X^-_0 = 0$, respectively. We define the coupled 2D grid variables $\{Y_{k,j} = (X_{k,j}^-,X_{k,j}^+) : k \in \N, j \in [k+1]\}$ and let the Markovian coupling be $\{Y_k = (Y_{k,0},\dots,Y_{k,k}) : k \in \N\}$. Recall that each edge $\textsf{\small BSC}(\delta)$ either copies its input bit with probability $1 - 2\delta$, or produces an independent $\Ber\!\left(\frac{1}{2}\right)$ bit with probability $2\delta$ (we will explicitly demonstrate this in the proof of Proposition \ref{Prop: Slow Growth of Layers} in Appendix \ref{Miscellaneous Proofs}). We couple $\{X^+_k : k \in \N\}$ and $\{X^-_k : k \in \N\}$ so that along any edge BSC of the 2D grid, say $(X_{k,j},X_{k+1,j})$, $X_{k,j}^+$ and $X_{k,j}^-$ are either both copied with probability $1 - 2\delta$, or a shared independent $\Ber\!\left(\frac{1}{2}\right)$ bit is produced with probability $2\delta$ that becomes the input to both $X_{k+1,j}^+$ and $X_{k+1,j}^-$. In other words, $\{X^+_k : k \in \N\}$ and $\{X^-_k : k \in \N\}$ ``run'' on the same 2D grid and have common BSCs. The Markovian coupling $\{Y_k : k \in \N\}$ has the following properties:
\begin{enumerate}
\item The ``marginal'' Markov chains are $\{X^+_k : k \in \N\}$ and $\{X^-_k : k \in \N\}$.
\item For every $k \in \N$, $X_{k+1}^+$ is conditionally independent of $X_{k}^-$ given $X_k^{+}$, and $X_{k+1}^-$ is conditionally independent of $X_{k}^+$ given $X_k^{-}$. 
\item For every $k \in \N$ and every $j \in [k+1]$, $X_{k,j}^+ \geq X_{k,j}^-$ almost surely--this is the monotonicity in the coupling.
\end{enumerate} 
The third property is straightforward to verify since $X_{0,0}^+ \geq X_{0,0}^-$ almost surely at the beginning, each edge BSC preserves monotonicity (whether it copies its input or generates a new shared bit), and the AND processing functions are monotonic in the sense that their outputs are non-decreasing when the number of $1$'s in their inputs increases. 

Since the marginal Markov chains $\{X^+_k : k \in \N\}$ and $\{X^-_k : k \in \N\}$ run on the same 2D grid with common BSCs, we keep track of the Markov chain $\{Y_k : k \in \N\}$ in a single coupled 2D grid. This 2D grid has the same underlying graph as the deterministic 2D grid described in subsection \ref{Deterministic 2D Grid Model}. Its nodes are the coupled 2D grid variables $\{Y_{k,j} = (X_{k,j}^-,X_{k,j}^+) : k \in \N, j \in [k+1]\}$, and we relabel the alphabet of these variables for simplicity. So, each $Y_{k,j}  = (X_{k,j}^-,X_{k,j}^+) \in \Y \triangleq \{0_{\sf{c}},1_{\sf{u}},1_{\sf{c}}\}$, where $0_{\sf{c}} = (0,0)$, $1_{\sf{u}} = (0,1)$, and $1_{\sf{c}} = (1,1)$. (Note that we do not require a letter $0_{\sf{u}} = (1,0)$ in this alphabet due to the monotonicity in the coupling.) Furthermore, each edge of the coupled 2D grid is a channel (conditional distribution) $W$ between the alphabets $\Y$ and $\Y$ that captures the action of a shared $\textsf{\small BSC}(\delta)$--we describe $W$ using the following row stochastic matrix:
\begin{equation}
\label{Eq: Coupled BSC} 
W = \bbordermatrix{
		& 0_{\sf{c}} & 1_{\sf{u}} & 1_{\sf{c}} \cr
      0_{\sf{c}} & 1-\delta & 0 & \delta \cr
      1_{\sf{u}} & \delta & 1-2\delta & \delta \cr
      1_{\sf{c}} & \delta & 0 & 1-\delta \cr }
\end{equation}
where the $(i,j)$th entry gives the probability of output $j$ given input $i$. It is straightforward to verify that $W$ describes the aforementioned Markovian coupling. Finally, the AND rule can be equivalently described on the alphabet $\Y$ as:
\begin{equation}
\label{Eq: Coupled And} 
\begin{array}{|c|c|c|}
\hline
x_1 & x_2 & x_1 \wedge x_2 \\
\hline
0_{\sf{c}} & * & 0_{\sf{c}} \\
1_{\sf{u}} & 1_{\sf{u}} & 1_{\sf{u}} \\
1_{\sf{u}} & 1_{\sf{c}} & 1_{\sf{u}} \\
1_{\sf{c}} & 1_{\sf{c}} & 1_{\sf{c}} \\
\hline
\end{array}
\end{equation}
where $*$ denotes any letter in $\Y$, and the symmetry of the AND rule covers all other possible input combinations. This coupled 2D grid model completely characterizes the Markov chain $\{Y_k : k \in \N\}$, which starts at $Y_0 = 1_{\sf{c}}$ almost surely. We next prove Theorem \ref{Thm: Deterministic And Grid} by further analyzing this model.  

\renewcommand{\proofname}{\bfseries \emph{Proof of Theorem \ref{Thm: Deterministic And Grid}}}

\begin{proof}
We first bound the TV distance between $P_{X_k}^+$ and $P_{X_k}^-$ using the maximal coupling characterization of TV distance:
$$ \left\|P_{X_k}^+ - P_{X_k}^-\right\|_{\sf{TV}} \leq \P\!\left(X_k^+ \neq X_k^- \right) = 1 - \P\!\left(X_k^+ = X_k^- \right) \, . $$
The events $\{X_k^+ = X_k^-\}$ are non-decreasing, i.e. $\{X_k^+ = X_k^-\} \subseteq \{X_{k+1}^+ = X_{k+1}^-\}$ for all $k \in 
\N$. Indeed, suppose for any $k \in \N$, the event $\{X_k^+ = X_k^-\}$ occurs. Since $\{X_k^+ = X_k^-\} = \{Y_k \in \{0_{\sf{c}},1_{\sf{c}}\}^{k+1}\} = \{\text{there are no } 1_{\sf{u}}\text{'s in level } k \text{ of the coupled 2D grid}\}$, \eqref{Eq: Coupled BSC} and \eqref{Eq: Coupled And} imply that there are no $1_{\sf{u}}$'s in level $k+1$. Hence, the event $\{X_{k+1}^+ = X_{k+1}^-\}$ occurs as well. Letting $k \rightarrow \infty$, we can use the continuity of $\P$ with the events $\{X_k^+ = X_k^-\}$ to get:
$$ \lim_{k \rightarrow\infty}{\left\|P_{X_k}^+ - P_{X_k}^-\right\|_{\sf{TV}}} \leq 1 - \lim_{k \rightarrow \infty}{\P\!\left(X_k^+ = X_k^- \right)} = 1 - \P\!\left(A\right) $$
where we define $A \triangleq \{\exists k \in \N, \, \text{there are no } 1_{\sf{u}} \text{'s in level } k \text{ of the coupled 2D grid}\}$. Therefore, it suffices to prove that $\P\!\left(A\right) = 1$.

To prove this, we recall a well-known result from \cite[Section 3]{2DOrientedPercolationDurrett} on \textit{oriented bond percolation} in two-dimensional lattices. Given the underlying DAG of our deterministic 2D grid from subsection \ref{Deterministic 2D Grid Model}, suppose we independently keep each edge ``open'' with some probability $p \in [0,1]$, and delete it (``closed'') with probability $1-p$. Let $\Omega_{\infty}$ be the event that there is an infinite open path starting at the root. Furthermore, let $R_k \triangleq \sup\{j \in [k+1] : \text{there is an open path from the root to the node } (k,j) \}$ and $L_k \triangleq \inf\{j \in [k+1] : \text{there is an open path from the root to the node } (k,j) \}$ be the rightmost and leftmost nodes at level $k \in \N$, respectively, that are connected to the root. (Here, we refer to the node $X_{k,j}$ using $(k,j)$ as we no longer associate a random variable to it.) It is proved in \cite[Section 3]{2DOrientedPercolationDurrett} that there exists a critical threshold $\delta_{\sf{perc}} \in [0,1]$ around which we observe the following phase transition phenomenon:
\begin{enumerate}
\item If $p > \delta_{\sf{perc}}$, then $\P_p(\Omega_{\infty}) > 0$ and:
\begin{equation}
\label{Eq: Rightmost and Leftmost Open Paths}
\P_p\!\left(\lim_{k \rightarrow \infty}{\frac{R_k}{k}} = \frac{1 + \alpha(p)}{2} \text{ and } \lim_{k \rightarrow \infty}{\frac{L_k}{k}} = \frac{1 - \alpha(p)}{2} \, \middle| \, \Omega_{\infty}\right) = 1 
\end{equation}
for some constant $\alpha(p) > 0$, where $\alpha(p)$ is defined in \cite[Section 3, Equation (6)]{2DOrientedPercolationDurrett}, and $\P_p$ is the probability measure defined by the bond percolation process. 
\item If $p < \delta_{\sf{perc}}$, then $\P_p(\Omega_{\infty}) = 0$.
\end{enumerate}
We will use this to prove $\P\!\left(A\right) = 1$ by considering two cases.

\textbf{Case 1:} Suppose $1 - 2\delta < \delta_{\sf{perc}}$ (i.e. $\delta > (1-\delta_{\sf{perc}})/2$) in our coupled 2D grid. The root of the coupled 2D grid is $Y_{0,0} = 1_{\sf{u}}$ almost surely, and we consider an oriented bond percolation process (as described above) with $p = 1 - 2\delta$. In particular, open edges correspond to BSCs that are copies (with probability $1 - 2\delta$). In this context, $\Omega_{\infty}^c$ is the event that the number of nodes connected to the root via a sequence of BSCs that are copies is finite. Suppose the event $\Omega_{\infty}^c$ occurs. Since \eqref{Eq: Coupled BSC} and \eqref{Eq: Coupled And} portray that a $1_{\sf{u}}$ moves from level $k$ to level $k+1$ only if one of its outgoing edges is open (and the corresponding BSC is a copy), there exists a level $k$ such that no node at level $k$ is a $1_{\sf{u}}$. This proves that $\Omega_{\infty}^c \subseteq A$. Using part 2 of the phase transition in oriented bond percolation, we get $\P(A) = 1$.

% TODO: Could write this more rigorously at the risk of being verbose.

\textbf{Case 2:} Suppose $1 - \delta > \delta_{\sf{perc}}$ (i.e. $\delta < 1-\delta_{\sf{perc}}$) in our coupled 2D grid. Consider an oriented bond percolation process (as described earlier) with $p = 1 - \delta$ that runs on the 2D grid, where an edge is open when the corresponding BSC is either copying or generating a $0$ as the new bit (i.e. this BSC takes a $0_{\sf{c}}$ to a $0_{\sf{c}}$, which happens with probability $1-\delta$ as shown in \eqref{Eq: Coupled BSC}). Let $B_k$ for $k \in \N\backslash\!\{0\}$ be the event that the BSC from $Y_{k-1,0}$ to $Y_{k,0}$ generates a new bit which equals $0$. Then, $\P(B_k) = \delta$ and $\{B_k:k \in \N\backslash\!\{0\}\}$ are mutually independent. So, the second Borel-Cantelli lemma tells us that infinitely many of the $B_k$'s occur almost surely. Furthermore, $B_k \subseteq \{Y_{k,0} = 0_{\sf{c}}\}$ for every $k \in \N\backslash\!\{0\}$. 

We next define the following sequence of random variables for all $i \geq 1$: 
\begin{align*}
L_i & \triangleq \min\!\left\{k \geq T_{i-1} + 1 : B_k \text{ occurs}\right\} \\
T_i & \triangleq 1 + \max\!\left\{k \geq L_i : \exists j \in [k+1], \, Y_{k,j} \text{ is connected to } Y_{L_i,0} \text{ by an open path}\right\}
\end{align*}
where we set $T_0 = 0$. Note that when $T_{i-1} = \infty$, we let $L_i = \infty$ almost surely. Furthermore when $T_{i-1} < \infty$, $L_i < \infty$ almost surely as infinitely many of the $B_k$'s occur almost surely. We also note that when $L_{i} < \infty$, the set $\{k \geq L_i : \exists j \in [k+1], \, Y_{k,j} \text{ is connected to } Y_{L_i,0} \text{ by an open path}\}$ is non-empty since $Y_{L_i,0}$ is always connected to itself, and $T_i - L_i - 1$ denotes the length of the longest open path connected to $Y_{L_i,0}$ (which could be infinity). Lastly, when $L_{i} = \infty$, we let $T_i = \infty$ almost surely. 

Let $\F_k$ for $k \in \N$ be the $\sigma$-algebra generated by the random variables $(Y_0,\dots,Y_k)$ and all the BSCs above level $k$ (where we include all events determining whether these BSCs are copies, and all events determining the independent bits they produce). Then, $\{\F_k : k \in \N\}$ is a filtration. It is straightforward to verify that $L_i$ and $T_i$ are \textit{stopping times} with respect to $\{\F_k : k \in \N\}$ for all $i \geq 1$. We can show this inductively. $T_0 = 0$ is trivially a stopping time, and if $T_{i-1}$ is a stopping time, then $L_i$ is clearly a stopping time. So, it suffices to prove that $T_i$ is a stopping time given $L_i$ is a stopping time. For any finite $m$, $\{T_i = m\}$ is the event that $L_i \leq m-1$ and the length of the longest open path connected to $Y_{L_i,0}$ is $m - 1 - L_i$. This event is contained in $\F_m$ because $L_i \leq m-1$ can be determined from $\F_{m-1} \subseteq \F_m$ (since $L_i$ is a stopping time), and the length of the longest path can be determined from $\F_m$ (rather than $\F_{m-1}$). Hence, $T_i$ is indeed a stopping time when $L_i$ is a stopping time.  

Now observe that:
\begin{align*}
\P(\exists k \geq 1, \, T_k = \infty) & = \P(T_1 = \infty) + \sum_{m = 2}^{\infty}{\P(\exists k \geq 2, \, T_k = \infty | T_1 = m) \P(T_1 = m)} \\
& = \P(T_1 = \infty) + \sum_{m = 2}^{\infty}{\P(\exists k \geq 1, \, T_k + m = \infty) \P(T_1 = m)} \\
& = \P(T_1 = \infty) + (1 - \P(T_1 = \infty)) \P(\exists k \geq 1, \, T_k= \infty) \, . 
\end{align*}
Here, $\P(\exists k \geq 1, \, T_k + m = \infty) = \P(\exists k \geq 2, \, T_k = \infty | T_1 = m)$ holds because the random variables $\{(L_i,T_i) : i \geq 2\}$ given $T_1 = m$ have the same distribution as $\{(L_{i-1} + m,T_{i-1} + m) : i \geq 2\}$; in particular, $L_i$ given $T_1 = m$ corresponds to $L_{i-1} + m$, and $T_i$ given $T_1 = m$ corresponds to $T_{i-1} + m$. Moreover, the conditioning on $\{T_1 = m\}$ can be removed because the event $\{T_1 = m\}$ is in $\F_m$ since $T_1$ is a stopping time, and $\{T_1 = m\}$ is therefore independent of the events determining when BSCs below level $m$ are open, as well as the events $B_k$ for $k > m$. Rearranging the previous equation, we get:
$$ \P(\exists k \geq 1, \, T_k = \infty) \P(T_1 = \infty) = \P(T_1 = \infty) \, . $$
Since $\P(T_1 = \infty) = \P(\Omega_{\infty}) > 0$ by part 1 of the phase transition in oriented bond percolation, we have:
\begin{equation}
\label{Eq: Infinite Path Existence}
\P(\exists k \geq 1, \, T_k = \infty) = 1 \, . 
\end{equation}
Let $\Omega_k^{\sf{left}}$, respectively $\Omega_k^{\sf{right}}$, be the event that there exists an infinite open path connected to $Y_{k,0}$, respectively $Y_{k,k}$, for $k \in \N$. If $\{\exists k \geq 1, \, T_k = \infty\}$ occurs, we can choose the smallest $m$ such that $T_m = \infty$, and for this $m$, there is an infinite open path connected to $Y_{L_m,0} = 0_{\sf{c}}$ (where $Y_{L_m,0} = 0_{\sf{c}}$ because $B_{L_m}$ occurs). Hence, using \eqref{Eq: Infinite Path Existence}, we have:
$$ \P\!\left(\exists k \geq \N, \, \{Y_{k,0} = 0_{\sf{c}}\} \cap \Omega_k^{\sf{left}}\right) = 1 \, . $$
Likewise, we can also prove that:
$$ \P\!\left(\exists k \geq \N, \, \{Y_{k,k} = 0_{\sf{c}}\} \cap \Omega_k^{\sf{right}}\right) = 1 $$
which implies that:
\begin{equation}
\label{Eq: Paths Meet}
\P\!\left(\exists k \geq \N, \exists m \in \N, \, \{Y_{k,0} = Y_{m,m} = 0_{\sf{c}}\} \cap \Omega_k^{\sf{left}} \cap \Omega_m^{\sf{right}} \right) = 1 \, .
\end{equation}
 
To finish the proof, consider $k,m \in \N$ such that $Y_{k,0} = Y_{m,m} = 0_{\sf{c}}$, and $\Omega_k^{\sf{left}}$ and $\Omega_m^{\sf{right}}$ both happen. Let $R_n^{\sf{left}} = \sup\{j \in [n+1] : \text{there is an open path from } Y_{k,0} \text{ to } Y_{n,j} \}$ be the rightmost node at level $n > k$ that is connected to $Y_{k,0}$ by an open path, and $L_n^{\sf{right}} = \inf\{j \in [n+1] : \text{there is an open path from } Y_{m,m} \,$ $\text{to } Y_{n,j} \}$ be the leftmost node at level $n > m$ that is connected to $Y_{m,m}$ by an open path. Using \eqref{Eq: Rightmost and Leftmost Open Paths}, we know that:
$$ \lim_{n \rightarrow \infty}{\frac{R_n^{\sf{left}}}{n}} = \lim_{n \rightarrow \infty}{\frac{R_n^{\sf{left}}}{n-k}} = \frac{1 + \alpha(1-\delta)}{2} \quad \text{and} \quad \lim_{n \rightarrow \infty}{\frac{L_n^{\sf{right}}}{n}} = \lim_{n \rightarrow \infty}{\frac{L_n^{\sf{right}} - m}{n-m}} = \frac{1 - \alpha(1-\delta)}{2} \, . $$
This implies that:
$$ \lim_{n \rightarrow \infty}{\frac{R_n^{\sf{left}} - L_n^{\sf{right}}}{n}} = \alpha(1-\delta) > 0 $$
which means that for some sufficiently large level $n$, the rightmost open path from $Y_{k,0}$ meets the leftmost open path from $Y_{m,m}$. By construction, all nodes in these two paths are equal to $0_{\sf{c}}$. Furthermore, since these two paths meet, we see from \eqref{Eq: Coupled BSC} and \eqref{Eq: Coupled And} (which show that AND gates and BSCs output $0_{\sf{c}}$'s or $1_{\sf{c}}$'s when the inputs are $0_{\sf{c}}$'s or $1_{\sf{c}}$'s) that every node at level $n$ must be equal to $0_{\sf{c}}$ or $1_{\sf{c}}$. Hence, there exists a level $n$ with no $1_{\sf{u}}$'s, i.e. the event $A$ occurs. Therefore, we get $\P(A) = 1$ using \eqref{Eq: Paths Meet}. 

Combining the two cases completes the proof as $\P(A) = 1$ for any $\delta \in \left(0,\frac{1}{2}\right)$.
% TODO: Could write this more rigorously at the risk of being verbose.
\end{proof}

\renewcommand{\proofname}{\bfseries \emph{Proof}}

We remark that this proof can be perceived as using the technique presented in \cite[Theorem 5.2]{MarkovMixing}. Indeed, let $T \triangleq \inf\{k \in \N : X_k^+ = X_k^-\}$ be a stopping time with respect to $\{Y_k : k \in \N\}$ denoting the first time that the marginal Markov chains $\{X_k^+ : k \in \N\}$ and $\{X_k^- : k \in \N\}$ meet. (Note that $\{T = \infty\}$ corresponds to the event that these chains never meet.) Since the events $\{X_k^+ = X_k^-\}$ for $k \in \N$ form a non-decreasing sequence of sets, $\{T > k\} = \{X_k^+ \neq X_k^-\}$. We can use this relation to obtain the following bound on the TV distance between $P_{X_k}^+$ and $P_{X_k}^-$:
\begin{equation}
\label{Eq: Mixing Time Technique}
\left\|P_{X_k}^+ - P_{X_k}^-\right\|_{\sf{TV}} \leq \P\!\left(X_k^+ \neq X_k^- \right) = \P\!\left(T > k\right) = 1 - \P\!\left(T \leq k\right) 
\end{equation}
where letting $k \rightarrow \infty$ and using the continuity of $\P$ produces:
\begin{equation}
\lim_{k \rightarrow\infty}{\left\|P_{X_k}^+ - P_{X_k}^-\right\|_{\sf{TV}}} \leq 1 - \P\!\left(\exists k \in \N, \, T \leq k\right) = 1 - \P\!\left(T < \infty\right) \, . 
\end{equation}
These bounds correspond to the ones shown in \cite[Theorem 5.2]{MarkovMixing}. Since the event $A = \{\exists k \in \N, \, T \leq k\} = \{T < \infty\}$, our proof that $A$ happens almost surely also demonstrates that the two marginal Markov chains meet after a finite amount of time almost surely.

\section{Analysis of Deterministic XOR 2D Grid}
\label{Analysis of Deterministic Xor Grid}

In this section, we will prove Theorem \ref{Thm: Deterministic Xor Grid}. We let $\mathbb{F}_2 = \{0,1\}$ denote the Galois field of order $2$ (i.e. integers with addition and multiplication modulo $2$), $\mathbb{F}_2^n$ with $n \geq 2$ denote the vector space over $\mathbb{F}_2$ of column vectors with $n$ entries from $\mathbb{F}_2$, and $\mathbb{F}_2^{m \times n}$ with $m,n \geq 2$ denote the space of $m \times n$ matrices with entries in $\mathbb{F}_2$. (All matrix and vector operations will be performed modulo $2$.) Now fix some matrix $H \in \mathbb{F}_2^{m \times n}$ that has the following block structure:
\begin{equation}
\label{Eq: Block Structure of PCM}
H = \left[ 
\begin{array}{cc}
1 & B_1 \\
0 & B_2 
\end{array} \right]
\end{equation}
where $0 = [0 \cdots 0]^T \in \mathbb{F}_2^{m-1}$ is the zero vector (whose dimension will be understood from context in the sequel), $B_1 \in \mathbb{F}_2^{1 \times (n-1)}$, and $B_2 \in \mathbb{F}_2^{(m-1)\times(n-1)}$. Consider the following two problems:
\begin{enumerate}
\item \textbf{Coding Problem:} Let $\C \triangleq \{x \in \mathbb{F}_2^{n} : H x = 0\}$ be the \textit{linear code} defined by the \textit{parity check matrix} $H$. Let $X = [X_1 \enspace X_2^T]^T$ with $X_1 \in \mathbb{F}_2$ and $X_2 \in \mathbb{F}_2^{n-1}$ be a codeword drawn uniformly from $\C$. Assume that there exists a codeword $x = [1 \enspace x_2^T]^T \in \C$ (i.e. $B_1 x_2 = 1$ and $B_2 x_2 = 0$). Then, $X_1$ is a $\Ber\!\left(\frac{1}{2}\right)$ random variable. We observe the codeword $X$ through an additive noise model and see $Y_1 \in \mathbb{F}_2$ and $Y_2 \in \mathbb{F}_2^{n-1}$:
\begin{equation}
\label{Eq: Coding Problem}
\left[ \begin{array}{c} Y_1 \\ Y_2 \end{array} \right] = X + \left[ \begin{array}{c} Z_1 \\ Z_2 \end{array} \right] = \left[ \begin{array}{c} X_1 + Z_1 \\ X_2 + Z_2 \end{array} \right]
\end{equation} 
where $Z_1 \in \mathbb{F}_2$ is a $\Ber\!\left(\frac{1}{2}\right)$ random variable, $Z_2 \in \mathbb{F}_2^{n-1}$ is a vector of i.i.d. $\Ber\!\left(\delta\right)$ random variables that are independent of $Z_1$, and both $Z_1,Z_2$ are independent of $X$. Our problem is to decode $X_1$ with minimum probability of error after observing $Y_1,Y_2$. 
\item \textbf{Inference Problem:} Let $X^{\prime} \in \mathbb{F}_2$ be a $\Ber\!\left(\frac{1}{2}\right)$ random variable, and $Z \in \mathbb{F}_2^{n-1}$ be a vector of i.i.d. $\Ber\!\left(\delta\right)$ random variables that are independent of $X^{\prime}$. Suppose we see the observations $S_1^{\prime} \in \mathbb{F}_2$ and $S_2^{\prime} \in \mathbb{F}_2^{m-1}$ through the model:
\begin{equation}
\label{Eq: Inference Problem}
\left[ \begin{array}{c} S_1^{\prime} \\ S_2^{\prime} \end{array} \right] = H \left[ \begin{array}{c} X^{\prime} \\ Z \end{array} \right] =  \left[ \begin{array}{c} X^{\prime} + B_1 Z \\ B_2 Z \end{array} \right] \, . 
\end{equation}
Our problem is to decode $X^{\prime}$ with minimum probability of error after observing $S_1^{\prime},S_2^{\prime}$. 
\end{enumerate}
The inference problem corresponds to our setting of reconstruction in the XOR 2D grid (as we will soon see). The next lemma illustrates that this inference problem is ``equivalent'' to the aforementioned coding problem (which admits simpler analysis).

\begin{lemma}[Equivalence of Problems]
\label{Lemma: Equivalence of Problems}
The minimum probabilities of error of the coding problem in \eqref{Eq: Coding Problem} and the inference problem in \eqref{Eq: Inference Problem} are equal. Moreover, if we couple the random variables in the two problems so that $X_1 = X^{\prime}$ and $Z_2 = Z$ almost surely (i.e. these variables are shared by the problems), $X_2$ is generated from a conditional distribution $P_{X_2|X_1}$ so that $X$ is uniform on $\C$, and $Z_1$ is generated independently, then we get $S_1^{\prime} = B_1 Y_2$ and $S_2^{\prime} = B_2 Y_2$ almost surely.
\end{lemma}
% Mention sufficient statistic.

\begin{proof}
We first show that the minimum probabilities of error for the two problems are equal. The inference problem has the following likelihoods for every $s_1^{\prime} \in \mathbb{F}_2$ and every $s_2^{\prime} \in \mathbb{F}_2^{m-1}$:
\begin{align*}
P_{S_1^{\prime},S_2^{\prime}|X^{\prime}}\!\left(s_1^{\prime},s_2^{\prime}\middle|0\right) & = \sum_{z \in \mathbb{F}_2^{n-1}}{P_Z(z) \I\!\left\{B_1 z = s_1^{\prime}, B_2 z = s_2^{\prime}\right\}}  \\
P_{S_1^{\prime},S_2^{\prime}|X^{\prime}}\!\left(s_1^{\prime},s_2^{\prime}\middle|1\right) & = \sum_{z \in \mathbb{F}_2^{n-1}}{P_Z(z) \I\!\left\{B_1 z = s_1^{\prime} + 1, B_2 z = s_2^{\prime}\right\}}
\end{align*}
and its prior is $X^{\prime} \sim \Ber\!\left(\frac{1}{2}\right)$. On the other hand, the coding problem has the following likelihoods for every $y_1 \in \mathbb{F}_2$ and every $y_2 \in \mathbb{F}_2^{n-1}$:
\begin{align*}
P_{Y_1,Y_2|X_1}(y_1,y_2|0) & = P_{Y_1|X_1}(y_1|0) P_{Y_2|X_1}(y_2|0) = \frac{1}{2} \sum_{x_2 \in \mathbb{F}_2^{n-1}}{P_{Y_2|X_2}(y_2|x_2) P_{X_2|X_1}(x_2|0)} \\
& = \frac{1}{2} \sum_{x_2 \in \mathbb{F}_2^{n-1}}{P_{Z_2}(y_2-x_2) \I\!\left\{B_1 x_2 = 0, B_2 x_2 = 0\right\} \frac{2}{|\C|}} \\
& = \frac{1}{|\C|} \sum_{z_2 \in \mathbb{F}_2^{n-1}}{P_{Z_2}(z_2) \I\!\left\{B_1 z_2 = B_1 y_2, B_2 z_2 = B_2 y_2\right\}} \\
P_{Y_1,Y_2|X_1}(y_1,y_2|1) & = \frac{1}{|\C|} \sum_{z_2 \in \mathbb{F}_2^{n-1}}{P_{Z_2}(z_2) \I\!\left\{B_1 z_2 = B_1 y_2 + 1, B_2 z_2 = B_2 y_2\right\}} 
\end{align*}
and its prior is $X_1 \sim \Ber\!\left(\frac{1}{2}\right)$. For the coding problem, define $S_1 \triangleq B_1 Y_2$ and $S_2 \triangleq B_2 Y_2$. Due to the Fisher-Neyman factorization theorem \cite[Theorem 3.6]{Statistics}, the form of the likelihoods demonstrates that $(S_1,S_2)$ is a sufficient statistic of $(Y_1,Y_2)$ for performing inference about $X_1$. 

Continuing in the context of the coding problem, define the set $\C^{\prime} = \{x \in \mathbb{F}_2^{n-1} : B_1 x = 0 , B_2 x = 0\}$ (which is also a linear code), and for any fixed $s_1 \in \mathbb{F}_2$ and $s_2 \in \mathbb{F}_2^{m-1}$, define the set $\mathcal{S}(s_1,s_2) = \{y_1 \in \mathbb{F}_2, y_2 \in \mathbb{F}_2^{n-1} : B_1 y_2 = s_1, B_2 y_2 = s_2\}$. If there exists $y_2 \in \mathbb{F}_2^{n-1}$ such that $B_1 y_2 = s_1$ and $B_2 y_2 = s_2$, then $\mathcal{S}(s_1,s_2) = \{y_1 \in \mathbb{F}_2, y_2^{\prime} = y_2 + y : y \in \C^{\prime}\}$, which means that $|\mathcal{S}(s_1,s_2)| = 2|\C^{\prime}| = |\C|$ (where the final equality holds because each vector in $\C^{\prime}$ corresponds to a codeword in $\C$ whose first letter is $0$, and we have assumed that there are an equal number of codewords with first letter $1$). Hence, for every $s_1 \in \mathbb{F}_2$ and every $s_2 \in \mathbb{F}_2^{m-1}$, the likelihood of $(S_1,S_2)$ given $X_1 = 0$ is:
\begin{align*}
P_{S_1,S_2|X_1}(s_1,s_2|0) & = \sum_{y_1 \in \mathbb{F}_2, y_2 \in \mathbb{F}_2^{n-1}}{P_{Y_1,Y_2|X_1}(y_1,y_2|0) \I\!\left\{B_1 y_2 = s_1, B_2 y_2 = s_2\right\}} \\
& = \frac{|\mathcal{S}(s_1,s_2)|}{|\C|} \sum_{z_2 \in \mathbb{F}_2^{n-1}}{P_{Z_2}(z_2) \I\!\left\{B_1 z_2 = s_1, B_2 z_2 = s_2\right\}} \\
& = \sum_{z_2 \in \mathbb{F}_2^{n-1}}{P_{Z_2}(z_2) \I\!\left\{B_1 z_2 = s_1, B_2 z_2 = s_2\right\}} \, .
\end{align*}
% Mention that the case $|\mathcal{S}(s_1,s_2)| = 0$ does not need to be handled separately.
Likewise, for every $s_1 \in \mathbb{F}_2$ and every $s_2 \in \mathbb{F}_2^{m-1}$, the likelihood of $(S_1,S_2)$ given $X_1 = 1$ is:
$$ P_{S_1,S_2|X_1}(s_1,s_2|1) = \sum_{z_2 \in \mathbb{F}_2^{n-1}}{P_{Z_2}(z_2) \I\!\left\{B_1 z_2 = s_1 + 1, B_2 z_2 = s_2\right\}} \, . $$
These likelihoods are exactly the same as the likelihoods for the inference problem we computed earlier. So, the sufficient statistic $(S_1,S_2)$ in the coding problem is equivalent to the observation $(S_1^{\prime},S_2^{\prime})$ in the inference problem in the sense that they are defined by the same probability model. As a result, the minimum probabilities of error in these formulations must be equal.

We now assume that the random variables in the two problems are coupled as in the lemma statement. To prove that $S_1^{\prime} = S_1$ and $S_2^{\prime} = S_2$ almost surely, observe that:
$$ \left[ \begin{array}{c} S_1 \\ S_2 \end{array} \right] = \left[ \begin{array}{c} B_1 Y_2 \\ B_2 Y_2 \end{array} \right] = \left[ \begin{array}{c} B_1 X_2 + B_1 Z_2 \\ B_2 X_2 + B_2 Z_2 \end{array} \right] = \left[ \begin{array}{c} X_1 + B_1 Z_2 \\ B_2 Z_2 \end{array} \right] = H \left[ \begin{array}{c} X_1 \\ Z_2 \end{array} \right] = \left[ \begin{array}{c} S_1^{\prime} \\ S_2^{\prime} \end{array} \right] $$
where the second equality uses \eqref{Eq: Coding Problem}, the third equality holds because $B_1 X_2 = X_1$ and $B_2 X_2 = 0$ since $X$ is a codeword, and the last equality uses \eqref{Eq: Inference Problem} and the fact that $X_1 = X^{\prime}$ and $Z_2 = Z$ almost surely. This completes the proof. (We note that this proof illustrates that $(S_1^{\prime},S_2^{\prime})$ in the inference problem is actually a sufficient statistic for the coding problem under the coupling in the lemma statement.)
\end{proof}

Recall that we are given a deterministic 2D grid where all Boolean processing functions with two inputs are the XOR rule, and all Boolean processing functions with one input are the identity rule, i.e. $f_{1}(x_1,x_2) = x_1 \oplus x_2$ and $f_{2}(x) = x$. We next prove Theorem \ref{Thm: Deterministic Xor Grid} using Lemma \ref{Lemma: Equivalence of Problems}.

\renewcommand{\proofname}{\bfseries \emph{Proof of Theorem \ref{Thm: Deterministic Xor Grid}}}

\begin{proof}
We first prove that the problem of decoding the root bit in the XOR 2D grid is captured by the inference problem defined in \eqref{Eq: Inference Problem}. Let $E_k$ denote the set of all directed edges in the 2D grid above level $k \in \N$. Furthermore, let us associate each edge $e \in E_k$ with an independent $\Ber(\delta)$ random variable $Z_e \in \mathbb{F}_2$. Since a $\textsf{\small BSC}(\delta)$ can be modeled as addition of an independent $\Ber(\delta)$ bit (in $\mathbb{F}_2$), the random variables $\{Z_e : e \in E_k\}$ define the BSCs of the 2D grid up to level $k$. Furthermore, each node at level $k \geq 1$ of the XOR 2D grid is simply a sum (in $\mathbb{F}_2$) of its parent nodes and the random variables on the edges between it and its parents. This provides a recursive formula for each node in terms of its parent nodes that can be unwound so that each node can be represented in terms of the root bit and all edge random variables: 
\begin{equation}
\label{Eq: Sum Representation of Nodes}
\forall k \geq 1, \forall j \in \{0,\dots,k\}, \enspace X_{k,j} =  \left(\binom{k}{j} \Mod{2}\right) X_{0,0} \, + \sum_{e \in E_k}{b_{j,e}^k Z_e} 
\end{equation}
where the coefficient of $X_{0,0}$ can be computed by realizing that the coefficients of the nodes in the ``grid above $X_{k,j}$'' (with $X_{k,j}$ as the root) are defined by the recursion of Pascal's triangle, and $b_{j,e}^k \in \mathbb{F}_2$ are some fixed coefficients. We do not require detailed knowledge of the values of $\{b_{j,e}^k : k \ge 1, 0 \leq j \leq k, e \in E_k\}$ (but they can also be evaluated if desired via straightforward counting). In the remainder of this proof, we will fix $k$ to be a power of $2$: $k = 2^m$ for $m \in  \N\backslash\!\{0\}$. Then, we have:
$$ \binom{k}{j} \equiv \binom{2^m}{j} \equiv \left\{
\begin{array}{ll}
1 \, , & j = 0,k \\
0 \, , & j = 1,\dots,k-1
\end{array}
\right. \Mod{2} $$
since by Lucas' theorem (see \cite{LucasTheoremFine}), the parity of $\binom{k}{j}$ is $0$ if and only if at least one of the digits of $j$ in base $2$ is strictly greater than the corresponding digit of $k$ in base $2$, and the base $2$ representation of $k = 2^m$ is $10\cdots0$ (with $m$ $0$'s). So, for each $k$, we can write \eqref{Eq: Sum Representation of Nodes} in the form:
\begin{equation}
\label{Eq: Xor Grid Matrix Version}
\left[ \begin{array}{c} X_{k,0} \\ X_{k,1} \\ \vdots \\ X_{k,k-1} \\ X_{k,k} \end{array} \right] = \underbrace{\left[ \begin{array}{cccc} 1 & \text{---} & b_{0,e}^k & \text{---} \\ 0 & \text{---} & b_{1,e}^k & \text{---} \\ \vdots & & \vdots & \\ 0 & \text{---} & b_{k-1,e}^k & \text{---} \\ 1 & \text{---} & b_{k,e}^k & \text{---} \end{array} \right]}_{\displaystyle{\triangleq H_k}} \left[ \begin{array}{c} X_{0,0} \\ \mid \\ Z_e \\ \mid \end{array} \right] 
\end{equation}
where $H_k \in \mathbb{F}_2^{(k+1)\times (|E_k|+1)}$ is a binary matrix whose rows are indexed by the nodes at level $k$ and columns are indexed by $1$ (first index) followed by the edges in $E_k$. The rows of $H_k$ are made up of the coefficients in \eqref{Eq: Sum Representation of Nodes}, and the vector on the right hand side of \eqref{Eq: Xor Grid Matrix Version} has first element $X_{0,0}$ followed by the random variables $\{Z_e : e \in E_k\}$ (indexed consistently with $H_k$). Our problem is to decode $X_{0,0}$ from the observations $(X_{k,0},\dots,X_{k,k})$ with minimum probability of error. Note that we can replace the last row of $H_k$ by the sum of the first and last rows of $H_k$ (a row operation) to get $H_k^{\prime}$, and correspondingly replace $X_{k,k}$ by $X_{k,0} + X_{k,k}$ in \eqref{Eq: Xor Grid Matrix Version} to get the equivalent formulation (which has the same minimum probability of error for decoding $X_{0,0}$):
\begin{equation}
\label{Eq: Row Operated Problem}
\left[ \begin{array}{c} X_{k,0} \\ X_{k,1} \\ \vdots \\ X_{k,k-1} \\ X_{k,0} + X_{k,k} \end{array} \right] = H_k^{\prime} \left[ \begin{array}{c} X_{0,0} \\ \mid \\ Z_e \\ \mid \end{array} \right] 
\end{equation}
where the equivalence follows from the fact that we only perform invertible operations. Since $H_k^{\prime}$ is of the form \eqref{Eq: Block Structure of PCM}, the problem in \eqref{Eq: Row Operated Problem} is exactly of the form of the inference problem in \eqref{Eq: Inference Problem}.

We next transform the XOR 2D grid problem in \eqref{Eq: Xor Grid Matrix Version} into a coding problem. By Lemma \ref{Lemma: Equivalence of Problems}, the inference problem in \eqref{Eq: Row Operated Problem} is equivalent to a coupled coding problem analogous to \eqref{Eq: Coding Problem}. In this coupled coding problem, we generate a codeword $W_k = [X_{0,0} \enspace \text{---} \, W_{e}^k \, \text{---} \, ]^T$ uniformly from the linear code $\C_k$ defined by the parity check matrix $H_k^{\prime}$, where the first element of the codeword is $X_{0,0}$ and the remaining elements are $\{W_{e}^k : e \in E_k\}$. We then observe $W_k$ through the additive noise channel model: 
\begin{equation}
\label{Eq: Xor Coding Problem}
Y_k = W_k + \left[ \begin{array}{c} Z_{0,0}^k \\ \mid \\ Z_{e} \\ \mid \end{array} \right]
\end{equation}
where $\{Z_e : e \in E_k\}$ are the BSC variables, and $Z_{0,0}^k$ is an independent $\Ber\!\left(\frac{1}{2}\right)$ random variable. Our goal is to decode the first bit of the codeword, $X_{0,0}$, with minimum probability of error after observing $Y_k$. Since row operations do not change the nullspace of a matrix, we can equivalently think of $\C_k$ as the linear code generated by the parity check matrix $H_k$. Moreover, without loss of generality, the ML decoder for $X_{0,0}$ based on $Y_k$ (which achieves the minimum probability of error) in the coding problem makes an error if and only if the ML decision rule for $X_{0,0}$ based on $(X_{k,0},\dots,X_{k,k})$ in the inference problem in \eqref{Eq: Xor Grid Matrix Version} makes an error. This is because the coupling from Lemma \ref{Lemma: Equivalence of Problems} ensures that $(X_{k,0},\dots,X_{k,k})$ is a sufficient statistic of $Y_k$ for $X_{0,0}$ in the coding problem. Therefore, it suffices to study the probability of error in ML decoding for the coding problem \eqref{Eq: Xor Coding Problem} due to Lemma \ref{Lemma: Equivalence of Problems}.\footnote{We should remark that the equivalence between problems \eqref{Eq: Row Operated Problem} and \eqref{Eq: Xor Coding Problem} requires the existence of a codeword of the form $[1 \, w^T]^T$, with $w \in \mathbb{F}_2^{|E_k|}$, in $\C_k$ (as mentioned earlier). This condition is always satisfied. Indeed, such a codeword does not exist if and only if the first column of $H_k^{\prime}$ (which is $[1 \, 0 \cdots 0]^T$) is not in the span of the remaining columns of $H_k^{\prime}$. So, if such a codeword does not exist, we can decode $X_{0,0}$ in the setting of \eqref{Eq: Row Operated Problem} with zero probability of error because the observation vector on the left hand side of \eqref{Eq: Row Operated Problem} is in the span of the second to last columns of $H_k^{\prime}$ if and only if $X_{0,0} = 0$. (It is worth mentioning that in the coding problem in \eqref{Eq: Xor Coding Problem}, if such a codeword does not exist, we can also decode the first codeword bit with zero probability of error because all codewords must have the first bit equal to $0$.) Since it is clear that we cannot decode the root bit with zero probability of error in the XOR 2D grid, such a codeword always exists.}

Recall that each $\textsf{\small BSC}(\delta)$ copies its input with probability $1-2\delta$ and generates an independent $\Ber\!\left(\frac{1}{2}\right)$ bit with probability $2\delta$ (as shown in the proof of Proposition \ref{Prop: Slow Growth of Layers} in Appendix \ref{Miscellaneous Proofs}). Suppose we know which BSCs among $\{Z_e : e \in E_k\}$ generate independent bits in \eqref{Eq: Xor Coding Problem}. Then, we can perceive each BSC in $\{Z_e : e \in E_k\}$ as an independent \textit{binary erasure channel} (BEC) with erasure probability $2\delta$, denoted $\textsf{\small BEC}(2\delta)$, which erases its input if the corresponding BSC generates an independent bit, and copies its input otherwise. (Note that the BSC defined by $Z_{0,0}^k$ corresponds to $\textsf{\small BEC}(1)$ which always erases its input.) We now consider observing the codeword $W_k$ under this BEC model, where $X_{0,0}$ is erased almost surely, and the remaining bits of $W_k$ are erased independently with probability $2\delta$. The minimum probability of error in inferring $X_{0,0}$ using an ML decoder for this BEC model lower bounds the minimum probability of error in inferring $X_{0,0}$ using an ML decoder under the BSC model. This is clear because the BECs tell us which BSCs are generating independent bits, thereby providing additional information. Moreover, the fact that a $\textsf{\small BEC}(2\delta)$ is ``less noisy'' than a $\textsf{\small BSC}(\delta)$ is well-known in information theory, cf. \cite[Section 6, Equation (16)]{GraphSDPI}. Now let $I_k \subseteq E_k$ denote the set of indices where the corresponding elements of $W_k$ are \textit{not} erased. It is a standard exercise to show that there exists a codeword $w \in \C_k$ with first element $w_1 = 1$ and $w_i = 0$ for all $i \in I_k$ if and only if the ML decoder (for the BEC model) cannot recover $X_{0,0}$ and has probability of error $\frac{1}{2}$; see the discussion in \cite[Section 3.2]{ModernCodingTheory}. We next find a codeword with these properties when two particular erasures occur. 

Let $e_1 \in E_k$ and $e_2 \in E_k$ denote the edges $(X_{k-1,0},X_{k,0})$ and $(X_{k-1,k-1},X_{k,k})$ in the 2D grid, respectively. Consider the vector $\omega^k \in \mathbb{F}_2^{|E_k| + 1}$ such that $\omega^k_1 = 1$ (first bit is $1$), $\omega^k_{e_1} = \omega^k_{e_2} = 1$, and all other elements are $0$. Then, $\omega^k \in \C_k$ because:
$$ H_k \omega^k = \left[ \begin{array}{cccc} 1 & \text{---} & b_{0,e}^k & \text{---} \\ 0 & \text{---} & b_{1,e}^k & \text{---} \\ \vdots & & \vdots & \\ 0 & \text{---} & b_{k-1,e}^k & \text{---} \\ 1 & \text{---} & b_{k,e}^k & \text{---} \end{array} \right] \omega^k = \left[ \begin{array}{c} 1 + b_{0,e_1}^k + b_{0,e_2}^k \\ b_{1,e_1}^k + b_{1,e_2}^k \\ \vdots \\ b_{k-1,e_1}^k + b_{k-1,e_2}^k \\ 1 + b_{k,e_1}^k + b_{k,e_2}^k \end{array} \right] = 0 $$ 
where we use the facts that $b_{0,e_1}^k = 1$, $b_{0,e_2}^k = 0$, $b_{k,e_1}^k = 0$, $b_{k,e_2}^k = 1$, and for any $0< j < k$, $b_{j,e_1}^k = 0$ and $b_{j,e_2}^k = 0$ (and the value of $b_{j,e_i}^k$ for $i = 0,1$ and $j = 0,\dots,k$ is determined by checking the dependence of node $X_{k,j}$ on the variable $Z_{e_i}$, which is straightforward because $e_i$ is an edge between the last two layers at the side of the 2D grid). Since $\omega^k$ has two $1$'s at indices $e_1$ and $e_2$ (besides the first bit), if the BECs corresponding to the indices $e_1$ and $e_2$ erase their inputs, the ML decoder (for the BEC model) will fail to recover $X_{0,0}$ with probability of error $\frac{1}{2}$. 

Hence, we define $B_k$ to be the event that the BECs corresponding to edges $e_1$ and $e_2$ at level $k$ erase their inputs (or equivalently, the BSCs at these edges generate independent bits), which has probability $\P(B_k) = (2\delta)^2$ for every $k$. The events $\{B_k : k = 2^m, m \in \N\backslash\!\{0\}\}$ are mutually independent since the BSCs are all independent. By the second Borel-Cantelli lemma, infinitely many of the $B_{2^m}$'s occur almost surely. So, letting $A_n \triangleq \bigcup_{m = 1}^{n}{B_{2^m}}$ for $n \in \N\backslash\!\{0\}$, the continuity of the underlying probability measure $\P$ yields $\lim_{n \rightarrow \infty}{\P(A_n)} = 1$. Let $\widehat{X}_{0,0}^k = \widehat{X}_{0,0}^k(X_{k,0},\dots,X_{k,k})$ denote the ML decoder for $X_{0,0}$ at level $k = 2^m$ under the original (BSC) model \eqref{Eq: Xor Grid Matrix Version}. From our discussion, we know that:
$$ \forall m \in \N\backslash\!\{0\}, \enspace \P\!\left(\widehat{X}_{0,0}^{2^m} \neq X_{0,0} \,\middle|\,B_{2^m} \right) = \frac{1}{2} = \P\!\left(\widehat{X}_{0,0}^{2^m} \neq X_{0,0} \,\middle|\,A_m \right) $$
where the second equality holds because the probability of error for ML decoding is $\frac{1}{2}$ at a given level if the probability of error for ML decoding at the previous level is $
\frac{1}{2}$. Finally, observe that:
\begin{align*}
\lim_{m \rightarrow \infty} \P\!\left(\widehat{X}_{0,0}^{2^m} \neq X_{0,0}\right) & = \lim_{m \rightarrow \infty} \P\!\left(\widehat{X}_{0,0}^{2^m} \neq X_{0,0} \,\middle|\,A_m \right) \P(A_m) + \P\!\left(\widehat{X}_{0,0}^{2^m} \neq X_{0,0} \,\middle|\,A_m^c \right) \P(A_m^{c}) \\
& = \lim_{m \rightarrow \infty} \P\!\left(\widehat{X}_{0,0}^{2^m} \neq X_{0,0} \,\middle|\,A_m \right) \\
& = \frac{1}{2} \, . 
\end{align*}     
This completes the proof since the above condition is equivalent to \eqref{Eq: Deterministic Impossibility of Reconstruction}.
\end{proof}

\renewcommand{\proofname}{\bfseries \emph{Proof}}

\appendices

\section{Miscellaneous Proofs}
\label{Miscellaneous Proofs}

\renewcommand{\proofname}{\bfseries \emph{Proof of Corollary \ref{Cor: Existence of Grids where Reconstruction is Possible}}}

\begin{proof}
This follows from applying the probabilistic method. Indeed, we know from Theorem \ref{Thm:Phase Transition in Random Grid with Majority Rule Processing} that given $\delta < \frac{1}{6}$, for the random DAG model with $d = 3$, $L_m = \omega(\log(m))$, and majority processing functions, there exist $\epsilon > 0$ and $K \in \N$ such that:
$$ \forall k \geq K, \enspace \P\!\left(\hat{S}_{k} \neq X_{0}\right) \leq \frac{1}{2} - 2\epsilon \, . $$
Now define $P_k(G) \triangleq \P(h_{\sf{ML}}^k(X_k,G) \neq X_{0}|G)$ for $k \in \N$ as the conditional probability that the ML decision rule based on the full $k$-layer state $X_k$ makes an error given the random DAG $G$, and let $E_k$ for $k \in \N$ be the set of all DAGs $\mathcal{G}$ with $d = 3$ and $L_m = \omega(\log(m))$ such that $P_k(\mathcal{G}) \leq \frac{1}{2} - \epsilon$. Observe that for every $k \geq K$:
\begin{align*}
\frac{1}{2} - 2\epsilon \geq \P\!\left(\hat{S}_{k} \neq X_{0}\right) & = \E\!\left[\P\!\left(\hat{S}_{k} \neq X_{0}\middle|G\right)\right] \\
& \geq \E\!\left[P_k(G)\right] \\
& = \E\!\left[P_k(G)\middle|G \in E_k\right] \P\!\left(G \in E_k\right) + \E\!\left[P_k(G)\middle|G \not\in E_k\right] \P\!\left(G \not\in E_k\right) \\
& \geq \E\!\left[P_k(G)\middle|G \not\in E_k\right] \P\!\left(G \not\in E_k\right) \\
& \geq \left(\frac{1}{2} - \epsilon\right) \P\!\left(G \not\in E_k\right)
\end{align*}
where the first and third lines follow from the law of total expectation, the second line holds because the ML decision rule minimizes the probability of error, the fourth line holds because the first term in the previous line is non-negative, and the final line holds because $G \not\in E_k$ implies that $P_k(G) > \frac{1}{2} - \epsilon$. Then, we have for every $k \geq K$:
$$ \P\!\left(G \in E_k\right) \geq \frac{2\epsilon}{1 - 2\epsilon} > 0 \, . $$
Since $\{E_k : k \in \N\}$ form a non-increasing sequence of sets (because $P_k(G)$ is non-decreasing in $k$), we get via continuity:
$$ \P\!\left(G \in \bigcap_{k \in \N}{E_k}\right) = \lim_{k \rightarrow \infty}{\P\!\left(G \in E_k\right)} \geq \frac{2\epsilon}{1 - 2\epsilon} > 0 $$
which means that there exists a DAG $\mathcal{G}$ with $d = 3$ and $L_m = \omega(\log(m))$ such that $P_k(\mathcal{G}) \leq \frac{1}{2} - \epsilon$ for all $k \in \N$. This completes the proof.
\end{proof}

\renewcommand{\proofname}{\bfseries \emph{Proof of Proposition \ref{Prop: Slow Growth of Layers}}}

\begin{proof} 
We first prove part 1, where we are given a fixed deterministic DAG $\mathcal{G}$. Observe that the BSC along each edge of this DAG produces its output bit by either copying its input bit exactly with probability $1-2\delta$, or generating an independent $\Ber\!\left(\frac{1}{2}\right)$ output bit with probability $2\delta$. This is because the BSC transition matrix can be decomposed as:
$$ \def\arraystretch{1.1} 
\left[\begin{array}{cc} 1-\delta & \delta \\ \delta & 1-\delta \end{array}\right] = (1-2\delta)\left[\begin{array}{cc} 1 & 0 \\ 0 & 1 \end{array}\right] + (2\delta) \left[\begin{array}{cc} \frac{1}{2} & \frac{1}{2} \\ \frac{1}{2} & \frac{1}{2} \end{array}\right] \, .
\def\arraystretch{1} 
$$
%This idea comes from Fortuin-Kasteleyn random cluster representations of Ising models \cite{percolation}. 
Now consider the events:
$$ A_k \triangleq \{\text{all $d L_k$ edges from level $k-1$ to level $k$ generate independent output bits}\} $$
for $k \in \N\backslash\!\{0\}$, which have probabilities $\P(A_k) = (2\delta)^{d L_k}$ since the BSCs on the edges are independent. These events are mutually independent (once again because the BSCs on the edges are independent). Since the condition on $L_k$ in the proposition statement is equivalent to:
$$ \exists K \in \N, \forall k \geq K, \enspace (2\delta)^{d L_k} \geq \frac{1}{k} \, , $$
we must have:
$$ \sum_{k = 1}^{\infty}{\P(A_k)} \geq \sum_{k = K}^{\infty}{(2\delta)^{d L_k}} \geq \sum_{k = K}^{\infty}{\frac{1}{k}} = +\infty \, . $$ 
The second Borel-Cantelli lemma then tells us that infinitely many of the events $\{A_k : k \in \N\backslash\!\{0\}\}$ occur almost surely, i.e. $\P\!\left(\bigcap_{m = 1}^{\infty} \bigcup_{k = m}^{\infty} A_{k}\right) = 1$. In particular, if we define $B_m \triangleq \bigcup_{k = 1}^{m} A_{k}$ for $m \in \N\backslash\!\{0\}$, then by continuity:
\begin{equation}
\label{Eq: Forget whp}
\lim_{m \rightarrow \infty}{\P\!\left(B_m\right)} = \P\!\left(\bigcup_{k = 1}^{\infty} A_{k}\right) = 1 \, .
\end{equation}
Finally, observe that:
\begin{align}
\lim_{m \rightarrow \infty} \P\!\left(h_{\sf{ML}}^m(X_m,\mathcal{G}) \neq X_0\right) & = \lim_{m \rightarrow \infty} \P\!\left( h_{\sf{ML}}^m(X_m,\mathcal{G}) \neq X_0 \middle| B_m\right) \P(B_m) \nonumber \\
& \quad \quad \quad \enspace + \P\!\left(h_{\sf{ML}}^m(X_m,\mathcal{G}) \neq X_0 \middle| B_m^{c}\right) \P(B_m^{c}) \nonumber \\
& = \lim_{m \rightarrow \infty} \P\!\left( h_{\sf{ML}}^m(X_m,\mathcal{G}) \neq X_0 \middle| B_m\right) \nonumber \\
& = \lim_{m \rightarrow \infty} \P\!\left( h_{\sf{ML}}^m(X_m,\mathcal{G}) \neq X_0 \middle| X_0 = 0, B_m\right) \P(X_0 = 0) \nonumber \\
& \quad \quad \quad \enspace + \P\!\left( h_{\sf{ML}}^m(X_m,\mathcal{G}) \neq X_0 \middle| X_0 = 1, B_m\right) \P(X_0 = 1) \nonumber \\
& = \lim_{m \rightarrow \infty} \frac{1}{2} \, \P\!\left( h_{\sf{ML}}^m(X_m,\mathcal{G}) = 1 \middle| B_m\right) + \frac{1}{2} \, \P\!\left( h_{\sf{ML}}^m(X_m,\mathcal{G}) = 0 \middle| B_m\right) \nonumber \\
& = \frac{1}{2}
\label{Eq: ML Decoder Fails}
\end{align}
where $h_{\sf{ML}}^m(\cdot,\mathcal{G}):\{0,1\}^{L_m} \rightarrow \{0,1\}$ denotes the ML decision rule at level $m$ based on $X_m$, the second equality uses \eqref{Eq: Forget whp}, the third equality uses the fact that $X_0$ is independent of $B_m$, and the fourth equality holds because $X_m$ is conditionally independent of $X_0$ given $B_m$. The condition in \eqref{Eq: ML Decoder Fails} is equivalent to \eqref{Eq: Deterministic Impossibility of Reconstruction}, which proves part 1.

To prove part 2, notice that part 1 yields:
$$ \lim_{k \rightarrow \infty}{\left\|P_{X_{k}|G}^+ - P_{X_{k}|G}^-\right\|_{\sf{TV}}} = 0 $$
which holds pointwise for every realization of the random DAG $G$. So, we can take expectations with respect to $G$ and apply the bounded convergence theorem to obtain part 2. This completes the proof. 
\end{proof}

\section*{Acknowledgments}

A. Makur would like to thank Dheeraj Nagaraj and Ganesh Ajjanagadde for stimulating discussions. 

% References section
\bibliographystyle{IEEEtranSA}
\bibliography{BroadcastRefs}

% That's all folks
\end{document}